\documentclass[11pt]{article}

\usepackage{longtable}
\usepackage{microtype}
\usepackage[load-configurations=version-1]{siunitx} 
\usepackage{siunitx}
\usepackage{hyperref}
\usepackage{cite}
\hypersetup{colorlinks,
            linkcolor=black,
            citecolor=blue,
            urlcolor=magenta,
            linktocpage,
            plainpages=false}
\usepackage[margin=1.5in]{geometry}

\usepackage{microtype}
\usepackage{graphicx}
\usepackage{subfigure}
\usepackage{multirow, makecell}
\usepackage{booktabs}

\usepackage{lmodern}
\usepackage[T1]{fontenc}
\usepackage{amsfonts}
\usepackage[utf8]{inputenc}
\usepackage[english]{babel}
\usepackage{amsthm}
\usepackage{algorithm}
\usepackage{algorithmic}
\usepackage{bbold}
\usepackage{microtype}
\usepackage{graphicx}
\usepackage{subfigure}
\usepackage{wrapfig}

\usepackage{amsmath}
\usepackage{bbm}

\DeclareMathOperator*{\argmax}{arg\,max}

\newtheorem{Theorem}{Theorem}

\newtheorem{Corollary}{Corollary}
\newtheorem{Definition}{Definition}
\newtheorem{Proposition}{Proposition}

\newcommand{\R}{\mathbb{R}} 
\renewcommand{\H}{\mathcal{H}} 
\newcommand{\Sp}{\mathcal{S}} 
\renewcommand{\P}{\mathbb{P}} 
\newcommand{\E}{\mathbb{E}} 
\renewcommand{\d}{\mathrm{d}} 
\newcommand{\M}{\mathcal{M}} 

\newcommand{\A}{\mathcal{A}} 

\begin{document}

\title{Interpretable Stein Goodness-of-fit Tests on Riemannian Manifolds
}
\author{Wenkai Xu$^1$ \and Takeru Matsuda$^{2}$}
\date{%
    $^1$Gatsby Computational Neuroscience Unit\\%
    $^2$
    Statistical Mathematics Unit, RIKEN Center for Brain Science\\[2ex]%
}
\maketitle

\begin{abstract}
In many applications, we encounter data on Riemannian manifolds such as torus and rotation groups.
Standard statistical procedures for multivariate data are not applicable to such data. 
In this study, we develop goodness-of-fit testing and interpretable model criticism methods for general distributions on Riemannian manifolds, including those with an intractable normalization constant.
The proposed methods are based on extensions of kernel Stein discrepancy, which are derived from Stein operators on Riemannian manifolds.
We discuss the connections between the proposed tests with existing ones and provide a theoretical analysis of their asymptotic Bahadur efficiency.
Simulation results and real data applications show the validity of the proposed methods.
\end{abstract}

\section{Introduction}

In many scientific and machine learning applications, data appear in the domains described by Riemannian manifolds.
For example, structures of proteins and molecules are described by a pair of angular variables, which is identified with a point on the torus \cite{singh02}.
In computer vision, the orientation of a camera is represented by a $3 \times 3$ rotation matrix, which gives rise to data on the rotation group \cite{song2009hilbert}.
Other examples include the orbit of a comet \cite{jupp1979maximum} and the vectorcardiogram data \cite{downs1972orientation}.
In addition, shape analysis \cite{ dryden2016statistical} and compositional data analysis \cite{pawlowsky2011compositional} also deal with complex data defined on Riemannian manifolds.
{Recently, \cite{klein2020torus} developed a graphical model on torus to analyze {phase coupling} between neuronal activities.}
Since the usual statistical procedures for Euclidean data are not applicable, many studies have developed statistical models and methods tailored for data on Riemannian manifolds \cite{chikuse2012statistics,mardia99,ley2017modern}.

Statistical models on Riemannian manifolds are often given in the form of unnormalized densities with a computationally intractable normalization constant.
For example, the Fisher distribution on the rotation group \cite{chikuse2012statistics,sei2013properties} is defined by
\begin{equation}\label{eq:fisher}
    p(X \mid \Theta) \propto \exp(\mathrm{tr}(\Theta^{\top} X)), 
\end{equation}
and its normalization constant is not given in closed form.
Statistical inference with such models can become computationally intensive due to the intractable normalization constant.
Thus, statistical methods on Riemannian manifolds that do not require computation of the normalization constant have been developed for several tasks such as parameter estimation \cite{mardia2016score} and sampling \cite{girolami2009riemannian,ma2015complete}.
However, goodness-of-fit testing or model criticism procedures for general distributions on Riemannian manifolds is not established, to the best of our knowledge.

Kernel Stein discrepancy (KSD) \cite{gorham2015measuring, ley2017stein} is a discrepancy measure between distributions 
based on Stein's method \cite{barbour2005introduction,chen2010} and reproducing kernel Hilbert space (RKHS) theory \cite{RKHSbook}. 
KSD provides a general procedure for 
goodness-of-fit testing that does not require computation of the normalization constant, and it has shown state-of-the-art performance in various scenarios including Euclidean data \cite{chwialkowski2016kernel,liu2016kernelized}, discrete data \cite{yang2018goodness}, point processes \cite{yang2019stein}, censored data \cite{tamara2020kernelized} and directional data \cite{xu2020stein}. 
In addition, by using the technique of optimizing test power \cite{gretton2012optimal,sutherland2016generative}, KSD-based testing procedures also enable extraction of distributional features to perform model criticism \cite{jitkrittum2017linear,jitkrittum2018informative, kanagawa2019kernel, jitkrittum2020testing}. 
We note that Stein's method has recently been extended to Riemannian manifolds and applied to numerical integration \cite{barp2018riemannian} and Bayesian inference \cite{liu2018riemannian}.

In this paper, we develop goodness-of-fit testing and interpretable model criticism methods for general distributions on Riemannian manifolds.
After briefly reviewing background topics, we first introduce several types of Stein operators on Riemannian manifolds by using Stokes' theorem. 
Then, we define manifold kernel Stein discrepancies (mKSD) based on them and propose goodness-of-fit testing procedures, which do not require computation of the normalization constant. 
{We also develop mKSD-based interpretable model criticism procedures.}
Theoretical comparisons of test performance in terms of Bahadur efficiency are provided, {and simulation results validate the claims}. 
Finally, we provide real data applications to demonstrate the usefulness of the proposed methods. 

\section{Background}\label{sec:background}
\subsection{Distributions on Riemannian Manifolds}

In this paper, we focus on distributions on a smooth Riemannian manifold $(\M,g)$, where $g$ is a {Riemannian metric} on $\M$\footnote{In this paper, $\M$ may have non-empty boundary $\partial M$.}.
See \cite{lee2018introduction} for details on Riemannian geometry.
Here, we give several examples that will be used in experiments.
Note that we define the probability density of each distribution by its Radon--Nikodym derivative with respect to the volume element of $(\M,g)$.

\paragraph{Torus}
Bivariate circular data $(x_1,x_2) \in [0,2\pi)^2$ can be viewed as data on the torus $\Sp_1 \times \Sp_1$, where we identify $(\cos x,\sin x) \in \Sp_1$ with $x \in [0,2\pi)$.
To describe dependence between circular variables, \cite{singh02} proposed the bivariate von-Mises distribution:
\begin{align}
{p}(x_1,x_2 \mid \xi) \propto & \exp (\kappa_1 \cos(x_1-\mu_1) + \kappa_2 \cos(x_2-\mu_2) \nonumber \\
&+ \lambda_{12} \sin(x_1-\mu_1) \sin(x_2-\mu_2) ), \label{bvM}
\end{align}
where $\xi=(\kappa_1,\kappa_2,\mu_1,\mu_2,\lambda_{12})$, $\kappa_1 \geq 0$, $\kappa_2 \geq 0$, $0 \leq \mu_1 <2\pi$ and $0 \leq \mu_2 <2\pi$.
Its normalization constant is not represented in closed form.
We will apply this model to wind direction data in Section \ref{sec:exp}.




\paragraph{Rotation group}
The rotation group 
$\operatorname{SO(m)}$ 
is defined as
\[
   \operatorname{ SO(m)} = \{ X \in \mathbb{R}^{m \times m} \mid X^{\top} X = I_m, \det X = 1 \},
\]
where $I_m$ is the $m$-dimensional identity matrix.
The Fisher distribution \cite{chikuse2012statistics,sei2013properties} on $\operatorname{SO(m)}$ is defined as
$$
    p(X \mid \Theta) \propto \exp(\mathrm{tr}(\Theta^{\top}X)),
$$
for which the normalization constant is not given in closed form.
We will apply this model to vectorcardiogram data in Section \ref{sec:exp}.



The goodness-of-fit testing for general distributions on Riemannian manifolds is not established, to the best of our knowledge.
For tests of uniformity, several methods have been proposed such as the Sobolev test \cite{chikuse2004test,gine1975invariant,jupp2008data}.
However, they are not readily applicable to general disributions.
Although there are a few methods applicable to general distributions \cite{jupp2005sobolev,jupp2018measures}, they require computation of the normalization constant, which is often computationally intensive.
In addition, existing testing procedures cannot be applied to perform interpretable model criticism \cite{jitkrittum2016interpretable,kim2016examples,lloyd2015statistical}, which would provide an intuitive clarification of the discrepancy between the model and data.

\subsection{Kernel Stein Discrepancy on $\mathbb{R}^d$}

Here, we briefly review the goodness-of-fit testing with kernel Stein discrepancy on 
$\mathbb{R}^{d}$.
See \cite{chwialkowski2016kernel,liu2016kernelized} for more detail.

Let $q$ be a smooth probability density on $\mathbb{R}^d$.
For a smooth function $\mathbf{f}=(f_1,\dots,f_d):\mathbb{R}^d \to \mathbb{R}^d$, the Stein operator $\mathcal{T}_q$ is defined by 
\begin{align}
\mathcal{T}_q \mathbf{f}(x)&=\sum_{i=1}^d \left( f_i(x) \frac{\partial}{\partial x^i} \log q(x) + \frac{\partial}{\partial x^i} f_i(x) \right).
\label{eq:steinRd}
\end{align}
From integration by parts on $\mathbb{R}^d$, we obtain the equality, i.e. the Stein's identity
${\E}_q [\mathcal{T}_q \mathbf{f}]=0,$
under mild regularity conditions.
Since Stein operator $\mathcal{T}_q$ depends on the density $q$ only through the derivatives of $\log q$, it does not involve the normalization constant of $q$, which is 
a useful property 
for dealing with unnormalized models \cite{hyvarinen2005estimation}. 

Let $\mathcal{H}$ be a reproducing kernel Hilbert space (RKHS) on $\mathbb{R}^d$ and $\mathcal{H}^d$ be its product.
By using Stein operator, kernel Stein discrepancy (KSD) \cite{gorham2015measuring, ley2017stein} between two densities $p$ and $q$ is defined as 
$$\textnormal{KSD}(p\|q) =\sup_{\|\mathbf{f} \|_{\mathcal{H}^d} \leq 1} \E_{p}[\mathcal{T}_q \mathbf{f}].$$ 

It is shown that $\mathrm{KSD}(p\|q) \geq 0$ and $\mathrm{KSD}(p\|q) = 0$ if and only if $p=q$ under mild regularity conditions \cite{chwialkowski2016kernel}.
Thus, KSD is a proper discrepancy measure between densities.
After some calculation, $\mathrm{KSD}(p\|q)$ is rewritten as
\begin{align}
\mathrm{KSD}^2(p\|q) = {\E}_{x,\tilde{x} \sim p} [h_q(x,\tilde{x})], \label{eq:KSDequiv}
\end{align}
where $h_q$ does not involve $p$.

Given samples  $x_1,\dots,x_n$ from \emph{unknown} density $p$ on $\mathbb{R}^d$,
an empirical estimate of $\mathrm{KSD}^2(p\|q)$ can be obtained by using Eq.\eqref{eq:KSDequiv} in the form of U-statistics, and this estimate is used to test the hypothesis $H_0: p=q$, where the critical value is determined by bootstrap.
In this way, a general method of non-parametric goodness-of-fit test on $\mathbb{R}^d$ is obtained, which does not require computation of the normalization constant.

\section{Stein Operators on $\M$}\label{sec:stein_operator}
In this section, we introduce several types of Stein operators for distributions on Riemannian manifolds by using Stokes' theorem. 
The operators are categorized via the order of differentials of the input functions\footnote{Note that this should be distinguished from the differentials of the (unnormalized) density functions.}. 

\subsection{Differential Forms and Stokes' Theorem}\label{subsec:stokes}
To derive Stein operators on Riemannian manifolds, we need to use differential forms and Stokes' theorem.
Here, we briefly introduce these concepts. 
For more detailed and rigorous treatments, see \cite{flanders, lee2018introduction,spivak2018calculus}.

Let $\M$ be a smooth $d$-dimensional Riemannian manifold and take its local coordinate system $x^1,\dots,x^d$.
We introduce symbols ${\rm d}x^1,\dots,{\rm d}x^d$ and an associative and anti-symmetric operation $\wedge$ between them called the wedge product: ${\rm d}x^i \wedge {\rm d}x^j = -{\rm d}x^j \wedge {\rm d}x^i$.
Note that ${\rm d}x^i \wedge {\rm d}x^i = 0$.
Then, a $p$-form $\omega$ on $M$ ($0 \leq p \leq d$) is defined as
\[
    \omega = \sum_{i_1 \cdots i_p} f_{i_1 \cdots i_p} {\rm d} x^{i_1} \wedge \dots \wedge {\rm d} x^{i_p},
\]
where the sum is taken over all $p$-tuples $\{i_1, \cdots, i_p \} \subset \{1,\dots,d \}$ and each $f_{i_1 \cdots i_p}$ is a smooth function on $\M$.
The exterior derivative ${\rm d} \omega$ of $\omega$ is defined as the $(p+1)$-form given by
\[
    {\rm d} \omega = \sum_{i_1 \cdots i_p} \sum_{i=1}^d \frac{\partial f_{i_1 \cdots i_p}}{\partial x^i} {\rm d} x^i  \wedge {\rm d} x^{i_1} \wedge \dots \wedge {\rm d} x^{i_p}.
\]
For another coordinate system $y^1,\dots,y^d$ on $\M$, the differential form is transformed by
$${\rm d} y^j = \sum_{i=1}^d \frac{\partial y^j}{\partial x^i} {\rm d} x^i.$$
The volume element is defined as the $d$-form given by
\begin{align*}
(\det g)^{1/2} {\rm d} x^1 \wedge \dots \wedge {\rm d} x^d,
\end{align*}
where $g=g(x^1,\dots,x^d)$ is the $d \times d$ matrix of the Riemannian metric with respect to $x^1,\dots,x^d$.

The integration of a $d$-form on a $d$-dimensional manifold is naturally defined like the usual integration on $\mathbb{R}^d$ and invariant with respect to the coordinate selection.
Correspondingly, the integration by parts formula on $\mathbb{R}^d$ is generalized in the form of Stokes' theorem.

\begin{Proposition}[Stokes' theorem]\label{thm:stoke's}
Let $\partial \M$ be the boundary of $\M$ and $\omega$ be a $(d-1)$-form on $\M$.
Then,
\begin{align*}
\int_\M \d\omega=\int_{\partial \M}\omega.
\end{align*}
\end{Proposition}

\begin{Corollary}\label{cor:stokes}
If $\partial \M$ is empty, then
$
\int_{\M} \d\omega=0
$
for any $(d-1)$-form $\omega$ on $\M$.
\end{Corollary}




\paragraph{Coordinate choice}
In the following, to facilitate the derivation as well as computation of Stein operators, we assume that there exists a coordinate system $\theta^1,\dots,\theta^d$ on $\M$ that covers $\M$ almost everywhere.
For example, spherical coordinates for the hyperspheres and torus, generalized Euler angles \cite[Section 2.5.1]{chikuse2012statistics} for the rotation groups, and Givens rotations \cite{pourzanjani2017general} for the Stiefel manifolds satisfy this assumption.

\subsection{First Order Stein Operator}\label{subsec:first_order}

For a smooth probability density $q$ on $\M$ and a smooth function $\mathbf{f} = (f^1,\dots,f^{d}):\M \to \mathbb{R}^d$,
define a function ${\mathcal{A}}^{(1)}_{q}\mathbf{f}:\M \to \mathbb{R}$ by
\begin{equation}\label{eq:stein1_ae_coord}
{\mathcal{A}}^{(1)}_{q} \mathbf f = \sum_{i=1}^{d} \left( \frac{\partial f^i}{\partial {\theta}^i} + f^i \frac{\partial}{\partial {\theta}^i} \log (qJ) \right),
\end{equation}
where $J = (\det g)^{1/2}$ is the volume element. 
We refer to ${\mathcal{A}}^{(1)}_{q}$ as the first order Stein operator.
Note that \cite{xu2020stein} utilized this operator for goodness-of-fit testing on hyperspheres.

\begin{Theorem}\label{thm:stein1_global_coord}
If $\partial \M$ is empty or $f^1,\dots,f^d$ vanish on $\partial M$, then
\begin{align*}
{\E}_q [{\mathcal{A}}^{(1)}_{q} \mathbf f] = 0.
\end{align*}
\end{Theorem}
If $\M$ is a closed manifold such as torus and rotation group, it does not have boundary by definition and thus the assumption of Theorem~\ref{thm:stein1_global_coord} holds.
If the boundary of $\M$ is non-empty, a discussion relevant to the assumption of Theorem~\ref{thm:stein1_global_coord} can be found in \cite{liu2019estimating}, which studies density estimation on truncated domains.
Note that the assumption of Theorem~\ref{thm:stein1_global_coord}  is similar to Assumption 4 in \cite{barp2018riemannian}.



\subsection{Second Order Stein Operator}
In the context of numerical integration on Riemannian manifolds, \cite{barp2018riemannian} introduced a different type of Stein operator ${\mathcal{A}}^{(2)}_{q}$, which we call the second order Stein operator.
Specifically, for a smooth probability density $q$ on $\M$ and a smooth function $\tilde{f}: \M \to \R$, define ${\mathcal{A}}^{(2)}_{q}\tilde{f}:\M \to \mathbb{R}$ by
\begin{align}\label{eq:stein2}
\mathcal{A}^{(2)}_q \tilde{f} = \sum_{ij} \left( g^{ij}\frac{\partial^2 \tilde{f}}{\partial \theta^i \partial \theta^j}+ g^{ij}\frac{\partial \tilde{f}}{\partial \theta^j}\frac{\partial \log qJ}{\partial \theta^i} 
\right)
\end{align}
where we denote the inverse matrix of $(g_{ij})$ by $(g^{ij})$ following the convention of Riemmanian geometry.


\begin{Proposition}
[Proposition 1 of \cite{barp2018riemannian}]\label{thm:stein2}
If $\partial \M$ is empty or $\tilde{f}$ vanishes on $\partial M$, then
\begin{align*}
{\E}_q [\mathcal{A}^{(2)}_{q} \tilde{f}] = 0.
\end{align*}
\end{Proposition}

Theorem~\ref{thm:stein2} follows from Theorem~\ref{thm:stein1_global_coord}, because the second order Stein operator in Eq.\eqref{eq:stein2} can be viewed as a special case of the first order Stein operator in Eq.\eqref{eq:stein1_ae_coord} with
\begin{equation}\label{eq:derivative_form}
f^i = \sum_j g^{ij} \frac{\partial \tilde{f}}{\partial \theta^j}.
\end{equation}
Similar form of the second order Stein operator in Eq.\eqref{eq:stein2} has been studied in \cite{liu2018riemannian} for Bayesian inference. 
On the other hand, 
\cite{le2020diffusion} arrives at a similar 
second order Stein operator on Riemannian manifolds via Feller diffusion process in the context of density approximation. 

\subsection{Zeroth Order Stein Operator}
For a smooth probability density $q$ on $\M$ and a function $h: \M \to \R$, define a function ${\mathcal{A}}^{(0)}_{q} h:\M \to \mathbb{R}$ by
\[
{\mathcal{A}}^{(0)}_q h = h - \mathbb{E}_q[h],
\]
which clearly satisfies $\E_q [\A^{(0)}_q h]= 0$. 
Since $\A^{(0)}_q$ does not involve any differentials, we call it the zeroth order Stein operator.  
Compared to the first and second order Stein operators, this operator requires the normalization constant of $q$, which is often computationally intractable for Riemannian manifolds.
We will {show later} that this operator corresponds to the maximum mean discrepancy \cite{gretton2007kernel}.

\section{Goodness-of-fit Tests on $\M$} 
\label{sec:gof}
In this section, we propose goodness-of-fit testing procedures for distributions on Riemannian manifolds based on Stein operators in the previous section. 

\subsection{Manifold Kernel Stein Discrepancies}
By using Stein operators introduced in the previous section, we extend kernel Stein discrepancy to distributions on Riemannian manifolds.

Let $\H$ be a RKHS on $\M$ with reproducing kernel $k$ and $\H^d$ be its product.
We define the manifold kernel Stein discrepancies (mKSD) of the first, second and zeroth order by
\begin{align*}
    \operatorname{mKSD}^{(1)}(p\|q) &= \sup_{\|\mathbf{f} \|_{\H^d} \leq 1} \mathbb{E}_{p}[\mathcal{A}^{(1)}_q \mathbf{f}],\\
    \operatorname{mKSD}^{(2)}(p\|q) &= \sup_{\|\tilde{f} \|_{\H} \leq 1} \mathbb{E}_{p}[\mathcal{A}^{(2)}_q \tilde{f}],\\
    \operatorname{mKSD}^{(0)}(p\|q) &= \sup_{\|h \|_{\H} \leq 1} \mathbb{E}_{p}[\mathcal{A}^{(0)}_q h],
\end{align*}
respectively.
We also define the Stein kernels of first, second and zeroth order by
\begin{align*}
h^{(1)}_q(x, \tilde{x})=\left\langle \A^{(1)}_q k(x,\cdot), \A^{(1)}_q k(\tilde{x},\cdot)\right\rangle_{\H^d}, \\
h^{(2)}_q(x, \tilde{x})=\left\langle \A^{(2)}_q k(x,\cdot), \A^{(2)}_q k(\tilde{x},\cdot)\right\rangle_{\H}, \\
h^{(0)}_q(x, \tilde{x})=\left\langle \A^{(0)}_q k(x,\cdot), \A^{(0)}_q k(\tilde{x},\cdot)\right\rangle_{\H},
\end{align*}
respectively.
Then, by algebraic manipulation, we obtain the following.

\begin{Theorem}\label{thm:quadratic_form}
If $p$ and $q$ are smooth densities on $\M$ and the reproducing kernel $k$ of $\H$ is smooth, then
\begin{equation}
    \operatorname{mKSD}^{(c)}(p\|q)^2 = \E_{x,\tilde{x}}[h_q^{(c)}(x,\tilde{x})]
\end{equation}
for $c=0,1,2$, where $x,\tilde{x} \sim p$ are independent.
\end{Theorem}

From Theorem~\ref{thm:quadratic_form}, we can estimate mKSD by using samples from $p$. 
This is an important property in goodness-of-fit testing.



The following theorem shows that mKSD is a proper discrepancy measure between distributions on Riemannian manifolds.
The proof is given in supplementary material.
Let $L(x) = (L_1(x),\dots,L_d)^{\top}\in \R^d$ with
$$    L_i(x)=\frac{\partial}{\partial \theta^i} \log \frac{q(x)}{p(x)}.
$$

\begin{Theorem}\label{thm:characteristic}
Let $p$ and $q$ be smooth densities on $\M$. 
Assume:  1) The kernel $k$ vanishes at $\partial \M$ and is compact universal in the sense of \textnormal{\cite[Definition 2 (ii)]{carmeli2010vector}}; 2)
$\E_{x,\tilde{x} \sim p} [h^{(c)}_q(x,\tilde{x})^2] < \infty$, for $c=0,1,2$; 
3) $\E_{p} \| L(x) \|^2 < \infty$. 
Then, $\operatorname{mKSD}^{(c)}(p\|q)\geq 0$ and $\operatorname{mKSD}^{(c)}(p\|q)=0$ if and only if $p=q$.
\end{Theorem}

{Note that different mKSD uses different RKHS as the space of test functions. 
With the $d$-dimensional vector valued RKHS $\H^d$, $\operatorname{mKSD^{(1)}}$ takes the supremum over a larger class of functions than $\operatorname{mKSD^{(2)}}$, capturing richer distribution features. 
Theoretical analysis in testing context will be presented in Section \ref{sec:efficiency}.}

\paragraph{Equivalence of \textbf{$\operatorname{{mKSD^{(0)}}}$} and MMD}
For a RKHS $\H$, the maximum mean discrepancy (MMD) \cite{gretton2007kernel} between $p$ and $q$ is defined by
\[
\operatorname{MMD}(p \| q)^2  = \|\mu_p - \mu_q\|^2_{\H},
\]
where $\mu_p$, $\mu_q$ are the kernel mean embeddings \cite{muandet2017kernel} of $p$ and $q$, respectively.
The following theorem shows that \textbf{$\operatorname{{mKSD^{(0)}}}$} is equivalent to MMD.

\begin{Theorem}\label{thm:mmd_equiv}
$$
\operatorname{mKSD}^{(0)}(p \| q) = \operatorname{MMD}(p \| q).
$$
\end{Theorem}

\begin{proof}
By definition, we have 
\begin{align*}
    \operatorname{mKSD}^{(0)}(p \| q) &= \sup_{\|h \|_{\mathcal{H}} \leq 1} \mathbb{E}_{p}[{\mathcal{A}}^{(0)}_q h] \\
    &= \sup_{\|h \|_{\mathcal{H}} \leq 1} (\mathbb{E}_{p}[h] - \mathbb{E}_{q}[h]).
\end{align*}
Hence, taking the supreme in closed form via reproducing property, we obtain 
$$
\operatorname{mKSD}^{(0)}(p \| q)^2 = \|\mu_p - \mu_q\|^2_{\H} = \operatorname{MMD}(p \| q)^2.
$$
\end{proof}

\subsection{Goodness-of-fit Testing with mKSDs}\label{sec:gof_procedure}
Here, we present procedures for testing $H_0:p=q$ with significance level $\alpha$ based on samples $x_1,\dots,x_n \sim p$.

From Theorem~\ref{thm:quadratic_form}, an unbiased estimate of mKSD can be obtained in the form of U-statistics \cite{lee90}:
\begin{equation}\label{eq:u-stats}
{{\operatorname{mKSD}}}_u^{(c)}(p\|q)^2=\frac{1}{n(n-1)}\sum_{i\neq j}h^{(c)}_{q}(x_{i},x_{j}).
\end{equation}
Its asymptotic distribution is obtained via U-statistics theory \cite{lee90,van2000asymptotic} as follows.  
We denote the convergence in distribution by $\overset{d}{\to}$.

\begin{Theorem}\label{thm:u-stat-asymptotic}
For $c=0,1,2$, the following statements hold.

1. Under $H_0: p = q$, 
\begin{equation}\label{eq:null-dist}
n \cdot {\operatorname{mKSD}}_u^{(c)}(p\|q)^2  \overset{d}{\to} \sum_{j=1}^{\infty} w^{(c)}_{j}(Z_{j}^{2}-1),
\end{equation}
where $Z_{j}$ are i.i.d. standard Gaussian random variables and $w^{(c)}_{j}$ are the eigenvalues of the Stein kernel $h^{(c)}_{q}(x,\tilde{x})$ under $p(\tilde{x})$:
\begin{equation}\label{eq:eigen_fun}
\int h^{(c)}_{q}(x,\tilde{x})\phi_{j}(\tilde{x})p(\tilde{x}){\rm d}\tilde{x} = w^{(c)}_{j}\phi_{j}(x),
\end{equation}
where 
$\phi_{j}(x) \neq 0$ is the non-trivial eigen-function.

2. Under $H_1: p\neq q$, 
\[
\sqrt{n}\cdot \left({\operatorname{mKSD}}_u^{(c)}(p\|q)^2 - \operatorname{mKSD}^{(c)}(p\|q)^2\right)\overset{d}{\to}\mathcal{N}(0,{\sigma_c}^{2}),
\]
where ${\sigma_c}^{2}=\mathrm{Var}_{x\sim p}[\E_{\tilde{x}\sim p}[h^{(c)}_{q}(x,\tilde{x})]]>0$. 
\end{Theorem}

Based on Theorem~\ref{thm:u-stat-asymptotic}, we propose two procedures for goodness-of-fit testing.

{\paragraph{Spectrum Test}
We can also directly approximate the null distribution in Eq.\eqref{eq:null-dist} by using the eigenvalues of the Stein kernel matrix \cite[Theorem 1]{gretton2009fast}.}
Specifically, let $M^{(c)}$ be the $n\times n$ Stein kernel matrix defined by $(M^{(c)})_{ij} = h^{(c)}_q(x_i,x_j)$ and $\widetilde{w}_1^{(c)},\dots,\widetilde{w}_n^{(c)}$ be its eigenvalues. 
Then, we define the bootstrap samples by
\begin{equation}\label{eq:spectrum-null}
S_t = \frac{1}{n} \sum_{j=1}^{n} \widetilde{w}_j^{(c)}(Z_{j,t}^2 - 1), 
\end{equation}
where each $Z_{j,t}$ is the standard Gaussian variable.

\paragraph{Wild-bootstrap Test}
We employ the wild-bootstrap test with the V-statistics \cite{chwialkowski2014wild}.
The test statisic is given by
\begin{equation}\label{eq:v-stats}
{{\operatorname{mKSD}}}_v^{(c)}(p\|q)^2=\frac{1}{n^2}\sum_{i, j}h_{q}^{(c)}(x_{i},x_{j}).
\end{equation}
To approximate its null distribution, we define the wild-bootstrap samples by
\begin{equation}\label{eq:bootstrap-null}
S_t = \frac{1}{n^2}\sum_{i,j} W_{i,t}W_{j,t}h_q^{(c)}(x_i,x_j), 
\end{equation}
where each $W_{i, t} \in \{-1,1\}$ is the Rademacher variable of zero mean and unit variance.

The testing procedure is outlined in Algorithm \ref{alg:wild}.
We adopt this algorithm in the following experiments due to its computational efficiency.

\begin{algorithm}[ht]
   \caption{mKSD test via wild-bootstrap}
   \label{alg:wild}
\begin{algorithmic}[1]
\renewcommand{\algorithmicrequire}{\textbf{Input:}}
\renewcommand{\algorithmicensure}{\textbf{Objective:}}
\REQUIRE~~\\
    samples $x_1,\dots,x_n \sim p$, 
    null density $q$ \\
    kernel function $k$,
    test size $\alpha$\\
    bootstrap sample size $B$\\
\ENSURE~~\\
Test $H_0: p=q$ versus $H_1: p\neq q$.
\renewcommand{\algorithmicensure}{\textbf{Test procedure:}}
\ENSURE~~\\
\STATE Compute the statistic ${{\textnormal{mKSD}}_v}^{(c)}(p\|q)^2$, Eq.\eqref{eq:u-stats}.
\FOR{$t=1:B$}
\STATE Sample Rademacher variables $W_{1,t},\dots,W_{n,t}$. 
\STATE Compute $S_t$ by Eq.\eqref{eq:bootstrap-null}.
\ENDFOR
\renewcommand{\algorithmicrequire}{\textbf{Output:}}
\STATE Determine the $(1-\alpha)$-quantile $\gamma_{1-\alpha}$ of $S_1,\dots,S_B$.
\REQUIRE~~\\
Reject $H_0$ if ${{\textnormal{mKSD}}_v}^{(c)}(p\|q)^2 > \gamma_{1-\alpha}$; otherwise do not reject $H_0$.
\end{algorithmic}
\vspace{-0.08cm}
\end{algorithm}




\paragraph{Kernel choice}
The performance of kernel-based testing is sensitive to the choice of kernel parameters.
We choose the kernel parameters by maximizing an approximation of the test power following \cite{gretton2012optimal,jitkrittum2016interpretable,sutherland2016generative}.
From Theorem \ref{thm:u-stat-asymptotic},
\[
D := \sqrt{n}\frac{ {\textnormal{mKSD}_u^2(p\|q)}-\textnormal{mKSD}^2(p\|q)}{\sigma_{c}}\overset{d}{\to}\mathcal{N}(0,1)
\]
under the alternative $H_{1}:p \neq q$.
Thus, for sufficiently large $n$, the test power is approximated as
$\mathrm{Pr}_{H_1}(n \cdot {\textnormal{mKSD}_u^{(c)}(p\|q)}^2>r)\approx \Phi \left(\sqrt{n}\frac{\textnormal{mKSD}^{(c)}(p\|q)^2}{\sigma_{c}} \right)$ \cite{sutherland2016generative}.
Thus, we choose the kernel parameters by maximizing an estimate of $\textnormal{mKSD}^2(p\|q)/{\sigma}_{c}$ \cite{jitkrittum2017linear}.

\section{Model Criticism  on $\M$}\label{sec:model_criticism}
Now, we propose mKSD-based model criticism procedures for distributions on Riemannian manifolds. 
When the proposed model does not fit the observed data well, understanding which part of the model misfit the data is of practical interest. 
{The model criticism study can be helpful to better understand the representative prototype \cite{kim2016examples}, to criticize prior assumptions in Bayesian settings \cite{lloyd2015statistical} or to help better training of generative models \cite{sutherland2016generative}. 
With kernel-based non-parametric testing, distributional features can be extracted in the form of \emph{test locations} to represent areas that \emph{``best distinguish''} distributions. 
The locations where two sample distributions differ the most via MMD are studied in \cite{jitkrittum2016interpretable} and the most ``mis-specified'' locations between samples and models via KSD are studied in \cite{jitkrittum2017linear}. 
Recently, \cite{seth2019model} studied the model criticism via latent space, which may intrinsically correspond to Riemannian manifold structures. 
Such setting can be an interesting application of our development. 
}

Let
$\textbf{s}_p(\cdot) = \E_{\tilde{x}\sim p}[\A_q^{(1)} k(\tilde{x},\cdot)] \in \mathcal H^d$.
We define the manifold Finite Set Stein Discrepancy (mFSSD) adpated from \cite{jitkrittum2017linear} by
\begin{equation}\label{eq:mFSSD}
\operatorname{mFSSD}^2 = \frac{1}{dJ}\sum_{i=1}^d\sum_{j=1}^J ({\textbf{s}}_p(\mathbf{v}_j))_i^2,
\end{equation}
which can be computed in linear time of sample size $n$.
Stein identity of $\textbf{s}_p(\cdot)$
ensures  $\operatorname{mFSSD}^2 = 0$ under $H_0$ with probability one \cite[Theorem 1]{jitkrittum2017linear}.
To perform model criticism, we extract some test locations that give a higher detection rate (i.e., test power) than others. 
We choose the test locations $V=\left\{ \mathbf{v}_{j}\right\}_{j=1}^{J}$ by maximizing the approximate test power:
\begin{equation}\label{eq:mFFSD_obj}
    V = \argmax_{\mathbf{v}} \frac{\mathrm{mFSSD^{2}}}{\tilde{\sigma}_{H_{1}}},
\end{equation}
where $\tilde{\sigma}_{H_{1}}$ is the variance of $\operatorname{mFSSD}^2$ under $H_1$. More details are shown in Proposition \ref{prop:fssd_asym} and \ref{prop:mfssd-power} in the supplementary material.

\section{Comparison between mKSD Tests}\label{sec:efficiency}
\paragraph{Bahadur efficiency}
From Theorem \ref{thm:u-stat-asymptotic}, mKSD tests are consistent against all alternatives.
Thus, to understand which mKSD test is more powerful than others, we investigated their \emph{Bahadur efficiency} \cite{bahadur1960stochastic}, which quantify how fast the p-value goes to zero under alternatives.
Here, to focus on the effect of the choice of Stein operator on test performance, we briefly present results for testing of uniformity on the circle $\Sp^1$ under the von-Mises distribution.
See supplementary material for more details.
The technique of the proof is adapted from \cite{jitkrittum2017linear}.


\begin{Theorem}\label{thm:efficienty_12}(Scaling shift in von-Mises distribution)
Let $x\ \in \Sp^1$, $q(x) \propto 1$ and  $p(x)\propto \exp{(\kappa u^{\top} x)}$. Choose the von-Mises kernel of the form $k(x,x') = \exp{(x^{\top}x')}$. 
Denote the approximate Bahadur efficiency between $\operatorname{mKSD}$ with first and second order  Stein operators as
$${\rm E}_{1,2}(\kappa):=\frac{c^{(\operatorname{mKSD}^{(1)})}(\kappa)}{c^{(\operatorname{mKSD}^{(2)})}(\kappa)},$$ 
where $\kappa>0$. 
Then ${\rm E}_{1,2}(\kappa) > 1$.
\end{Theorem}

Adapting  \cite[Theorem 5]{jitkrittum2017linear}, it suffices to show $\operatorname{mKSD}^{(1)}(p\|q)$ $\geq$ $ \operatorname{mKSD}^{(2)}(p\|q)$ and 
$$
\E_{x,\tilde{x} \sim q}[h_q^{(2)}(x,\tilde{x})^2] > \E_{x,\tilde{x} \sim q}[h_q^{(1)}(x,\tilde{x})^2] >0.
$$ 
See supplementary material for details.

We provide additional discussion on the relative test efficiencies with $\operatorname{mKSD}^{(0)}$ in the supplementary material. 
In general, since we cannot compute $\E_p$ in closed form, especially with unnormlized density, we need to perform the test with samples drawn from the null, where sampling error makes the $\operatorname{mKSD}^{(0)}$ test less 
efficient in overall \cite{jitkrittum2017linear,yang2019stein, xu2020stein}. 

\paragraph{Computational efficiency}
Since the Stein kernels $h_q^{(1)}$ and $h_q^{(2)}$ depend on $q$ only through the derivative of $\log q$, mKSD tests with the first and second order Stein operators do not require computation of the normalization constant of $q$.
This is a major computational advantage over existing goodness-of-fit tests on Riemannian manifolds.
While the computational cost of $\operatorname{mKSD}_u^{(1)}$ is $O(n^2d)$, that of $\operatorname{mKSD}^{(2)}$ is $O(n^2d^3)$ due to the computation of the metric tensor.

On the other hand, mKSD test of zeroth order is equivalent to testing whether two sets of samples are from the same distribution by using MMD \cite{gretton2007kernel}.
Namely, to test whether $x_1,\dots,x_n$ is from density $q$, we draw samples $y_1,\dots,y_m$ from $q$ and determine whether $x_1,\dots,x_n$ and $y_1,\dots,y_m$ are from the same distribution. 
This procudre requires to sample from the null distribution $q$ on Riemannian manifolds, which is computationally intensive in general. 
Note that the results in Theorem \ref{thm:u-stat-asymptotic} with $c=0$ replicate the asymptotic results for MMD \cite{gretton2007kernel}.

\paragraph{Choosing mKSD tests}
Overall, $\operatorname{mKSD}^{(1)}$ has its advantage in terms of having a large space of test functions with both asymptotic test efficiency and computational efficiency so that it is recommended to use when available. 
$\operatorname{mKSD}^{(2)}$ can be slightly easier to compute and parameterize in particular scenarios, although it may sacrifice test power and computational efficiency. 
$\operatorname{mKSD}^{(0)}$, or namely MMD test, is also applicable when it is possible to sample from the given unnormalized density model on Riemannian manifolds.


\section{Simulation Results}\label{sec:simulation}
In this section, we show the validity of the proposed mKSD tests by simulation on the rotation group $\operatorname{SO(3)}$.
We use the Euler angle \cite{chikuse2012statistics} as the coordinate system. 
The bootstrap sample size is set to $B=1000$.
The significance level is set to $\alpha=0.01$.
For the mKSD$^{(0)}$ test (MMD two-sample test), the number of samples from the null is set to be equal to the sample size $n$. 
We used the kernel $k(X,Y)=\exp (\eta \cdot \mathrm{tr}(X^{\top} Y))$, where the parameter $\eta$ was chosen by optimizing the approximate test power. 
The exponential-trace kernel $k(X,Y) = \exp(\eta\cdot \mathrm{tr}(X^{\top} Y))$ for the rotation group is compact universal.
To see this, we rewrite the kernel in the form analogous to the Gaussian kernel:
$k(X,Y) = \exp(\eta\cdot \mathrm{tr}(X^{\top} Y)) = C \cdot \exp(-\frac{1}{2}\eta\cdot \|X - Y\|_F^2)$, where $C$ is a constant that only depends on $d$, the dimension of the matrices $X,Y \in \operatorname{SO(d)}$ and we know that $\mathrm{tr}(X^{\top} X)=\mathrm{tr}(I_d)=d$ for all $X \in \operatorname{SO(d)}$. 
Then, since the Gaussian kernel is universal \cite{sriperumbudur2011universality} and the rotation group $\operatorname{SO(d)}$ is a compact subset of the space of $d \times d$ matrices, the exponential-trace kernel is also compact-universal from Corollary 3 of \cite{carmeli2010vector}.


\subsection{Uniform distribution} 
First, we consider testing of uniformity on $\operatorname{SO(3)}$ and compare the performance of the mKSD tests with the Sobolev test \cite{jupp2005sobolev}.
We generated samples from the exponential trace distribution $p(X \mid \kappa) \propto \exp (\kappa \cdot \mathrm{tr}(X))$ by the rejection sampling \cite{hoff2009simulation}. 
The uniform distribution corresponds to $\kappa=0$. 


Figure \ref{fig:synthetic-problems} (a) plots the rejection rates with respect to $\kappa$ for $n=100$.
When $\kappa=0$, the type-I errors of all tests are well controlled to the significance level $\alpha = 0.01$.
The power of all tests increases with increasing $\kappa$ and converges to one.
Figure~\ref{fig:synthetic-problems} (b) plots the rejection rates with respect to $n$ for $\kappa=0.35$.
The power of all tests increases with $n$ and converges to one.
When the model becomes increasingly different from the null, the mKSD1 is more sensitive to distinguish the difference, with higher power than others. 
\subsection{Fisher distribution}
Next, we consider the {Fisher distribution} (or matrix-Langevin distribution) $p(X \mid F) \propto \exp (\mathrm{tr}(F^{\top} X))$ \cite{chikuse2003concentrated,sei2013properties}.
We generated data from $p(X \mid F_0)$ and applied mKSD tests on the null $p(X \mid F_b)$, where
$$
F_b = \begin{pmatrix}
1 & b & 0\\
b & 1 & 0\\
0 & 0 & 1
\end{pmatrix}.
$$
We compare the mKSD tests with the extended Sobolev test \cite{jupp2005sobolev}, in which we compute the normalization constant by Monte Carlo. 

Figure~\ref{fig:synthetic-problems} (c) plots the rejection rates with respect to $b$ for $n=100$.
Figure~\ref{fig:synthetic-problems} (d) plots the rejection rates with respect to $n$ for $b=0.2$.
From the plot, we see that all tests achieves the correct test level under the null. 
When the model becomes increasingly different from the null, the mKSD1 is more sensitive to distinguish the difference, with higher power than others. 
MMD test has lower power than mKSD1 and mKSD2 due to inefficiency from sampling. 
While the Sobolev test is useful when the null and the alternative are very different, 
it is not powerful for harder problems where the alternative perturbed very little from the null.

\begin{figure*}[t!]
	\begin{center}
        \subfigure[$n=100$]
		{\includegraphics[width=0.245\textwidth]{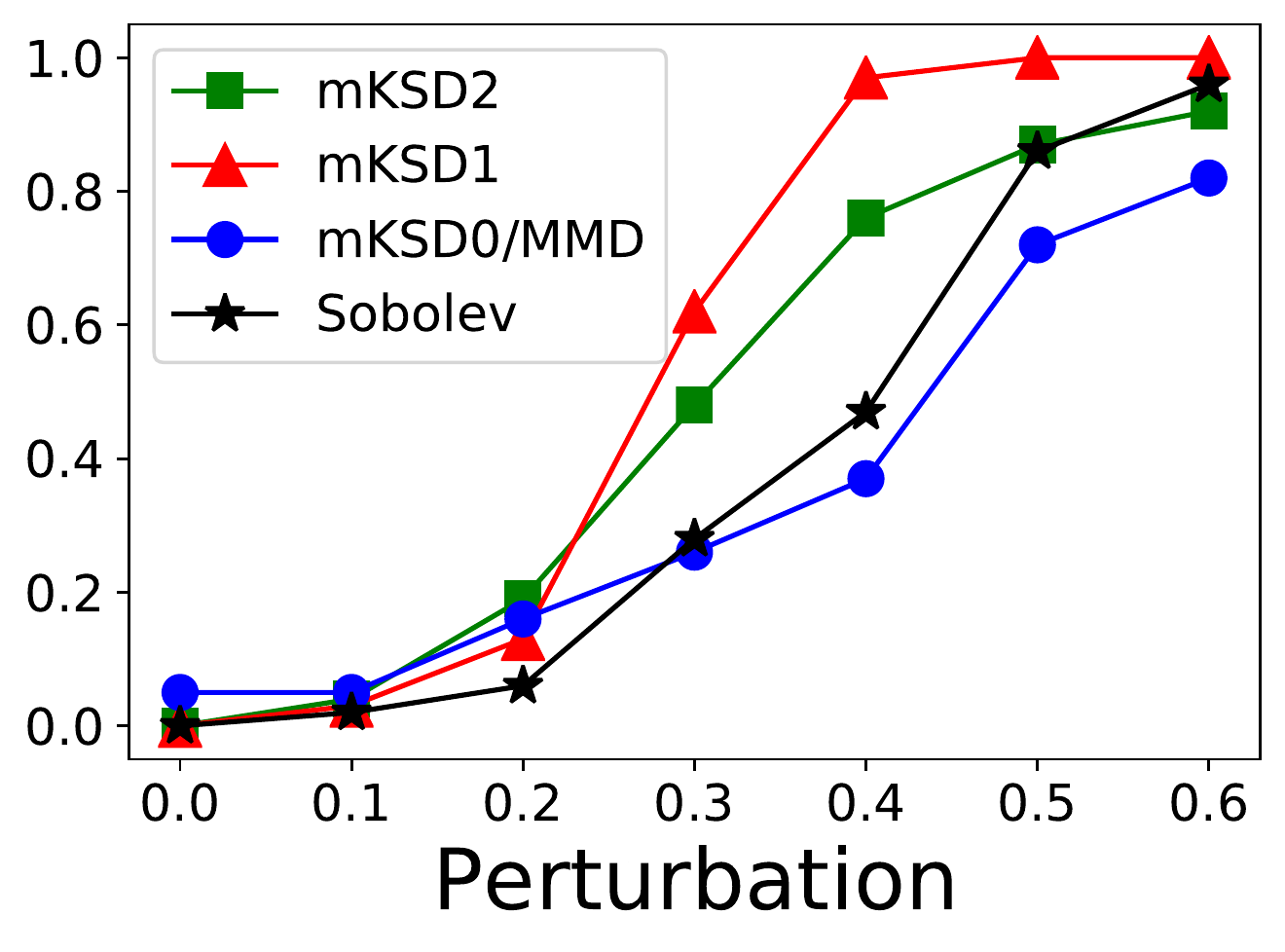}\label{fig:uniform-perturb}}\subfigure[$\kappa=0.35$]
		{\includegraphics[width=0.245\textwidth]{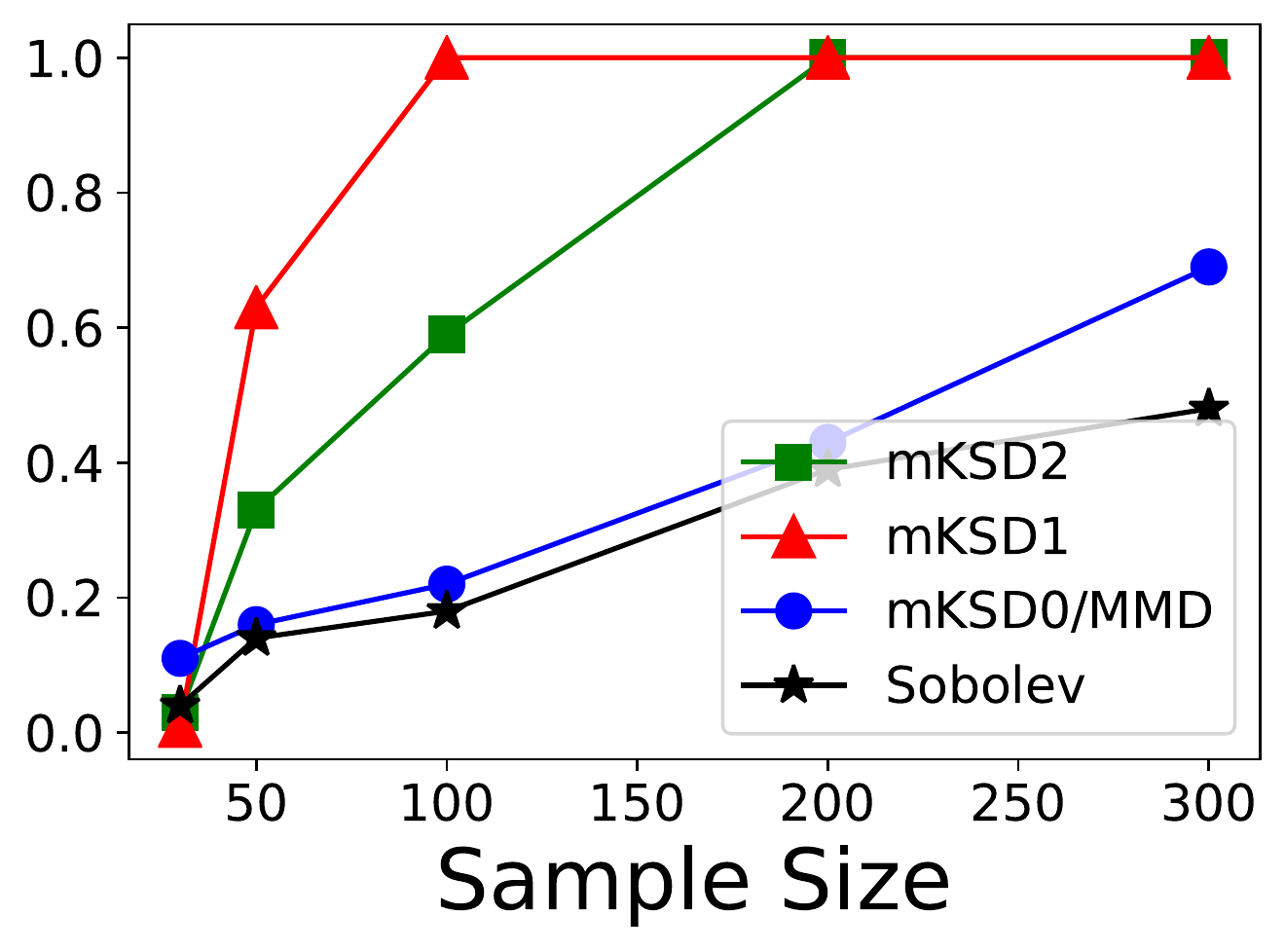}\label{fig:uniform-size}}
        \subfigure[$n=100$]
		{\includegraphics[width=0.245\textwidth]{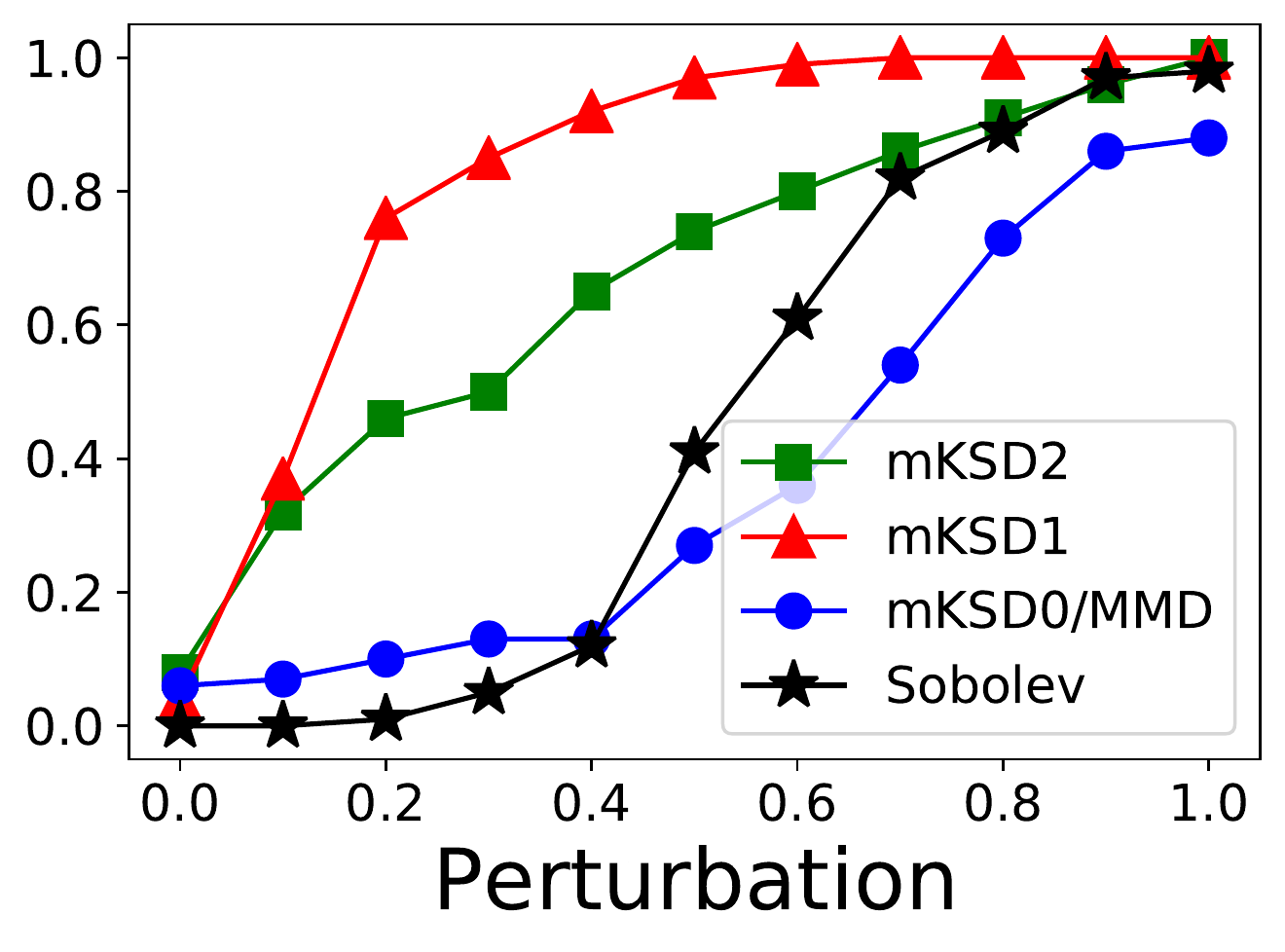}\label{fig:vmf-perturb}}\subfigure[$b=0.20$]
		{\includegraphics[width=0.245\textwidth]{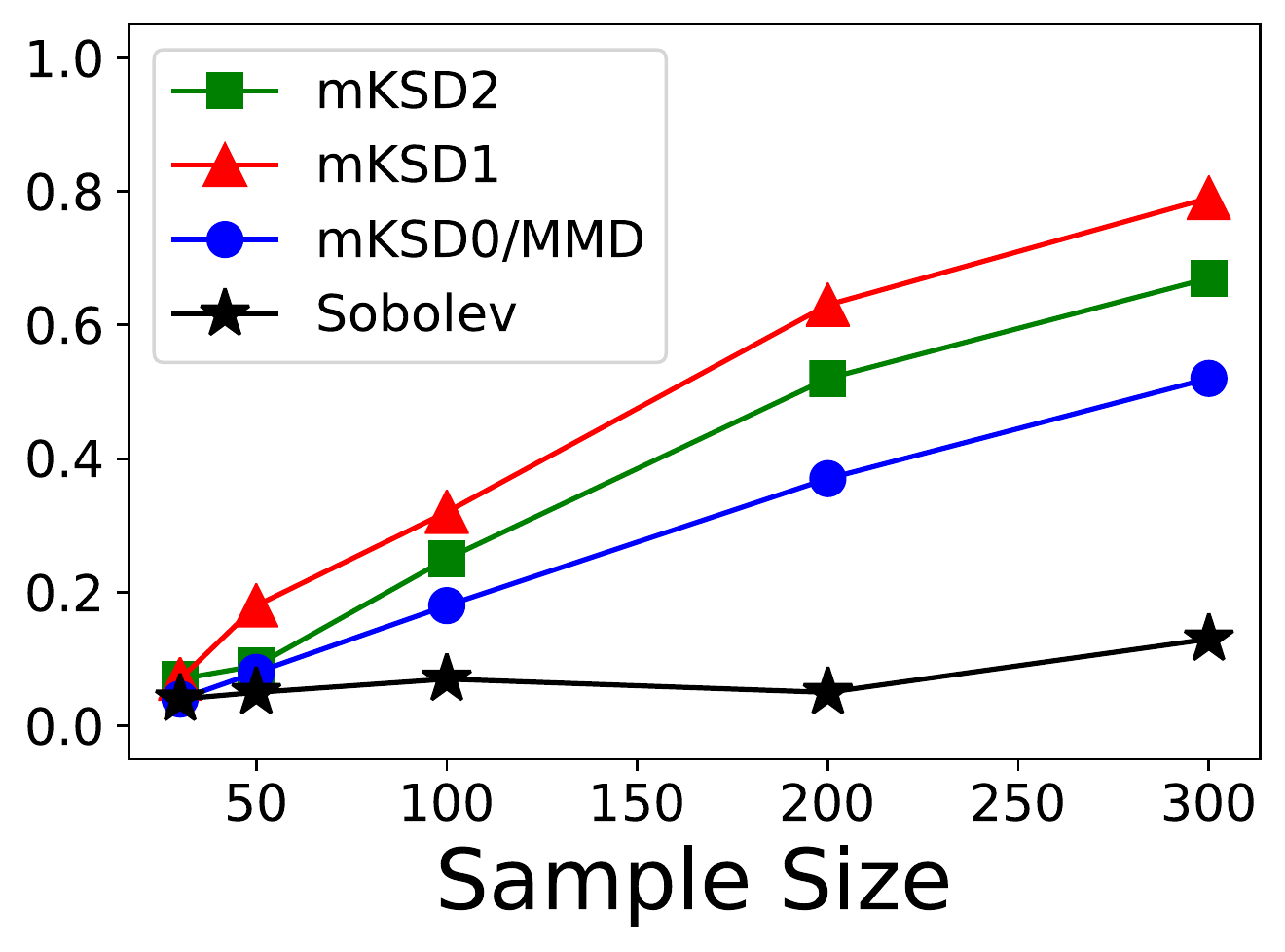}\label{fig:vmf-size}}
		\caption{Rejection rates at $\alpha=0.01$: a)-b) for uniform density; c)-d) for Fisher distribution on SO(3)}\label{fig:synthetic-problems}
	\end{center}
	\vspace{-0.3cm}
\end{figure*}

		

		

\section{Real Data Applications}\label{sec:exp}
Finally, we apply the mKSD tests to two real data problems.

\subsection{Vectorcardiogram data}
As a real dataset on the rotation group $\operatorname{SO(3)}$, we use the vectorcardiogram data studied by \cite{jupp2008data}. 
The data summarizes vectorcardiogram from normal children where each data point records 3 perpendicular vectors of directions QRS, PRS and T from Frank system for electrical lead placement. 
Details of this dataset can be found in \cite{downs1972orientation}. 
\cite{jupp2005sobolev} fitted the Fisher distribution $p(X \mid F) \propto \exp (\mathrm{tr}(F^{\top} X))$ to 28 data points of children aged between 2 to 10 and obtained the estimate
\[
\hat{F} = 5.63\times \begin{pmatrix}
0.583 & 0.629 & 0.514\\
0.660 & -0.736 & 0.151\\
0.473 & 0.252 & -0.844
\end{pmatrix}.
\]
We use this value as the null model to be tested. 
Table~\ref{tab:vcg} presents the p-values of each test\footnote{$k(X,Y)=\exp (\eta \cdot \mathrm{tr}(X^{\top} Y))$ is used as in Section \ref{sec:simulation}.
}.
All mKSD tests show strong evidence to reject the fitted model at 
$\alpha=0.05$; however, Sobolev test, with p-value=0.126, is not powerful enough to reject the null at the same test level.

\begin{table}[h]
    \centering\caption{p-values for vectorcardiogram data.}
    \label{tab:vcg}
\begin{tabular}{|c|c|c|c|}
\hline
mKSD1 &  mKSD2 &  mKSD0/MMD  &  Sobolev \\
\hline
0.004 & 0.000 &  0.010   &   0.126 \\
\hline
\end{tabular}
\end{table}

\begin{figure}[t!]
    \centering
    \includegraphics[width=0.3\textwidth,height=0.256\textwidth]{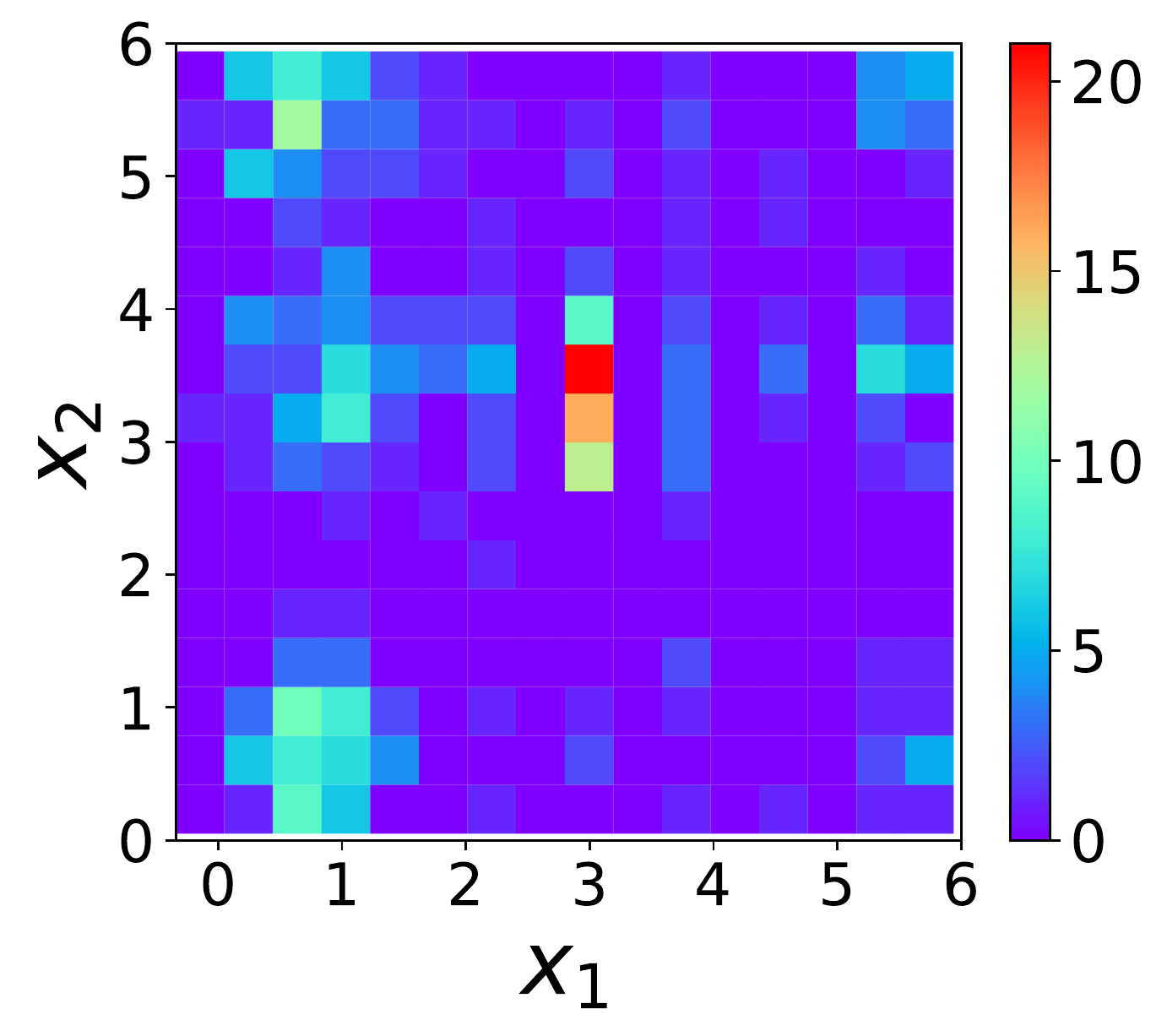}
        \includegraphics[width=0.25\textwidth,height=0.25\textwidth]{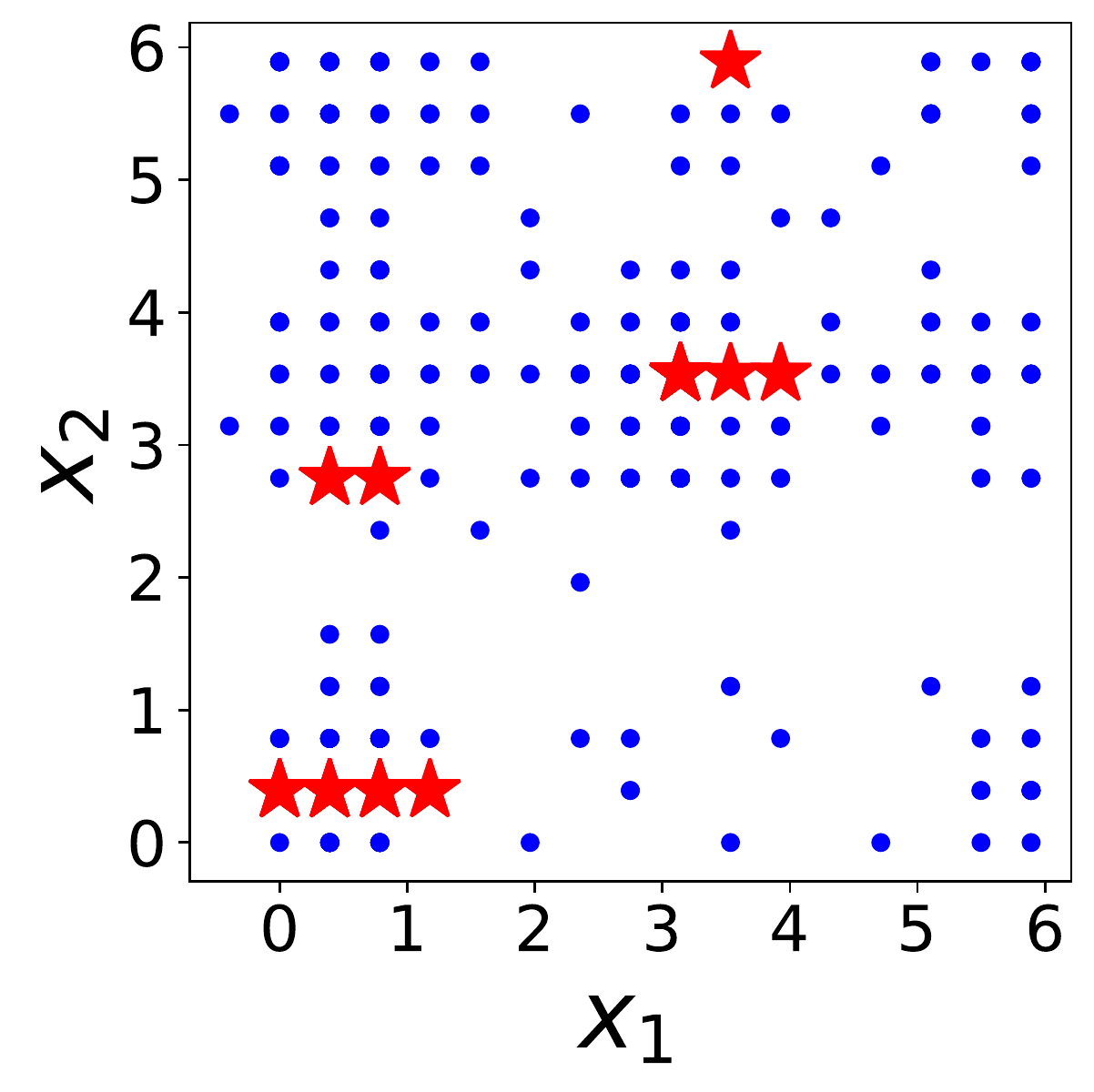}
            \includegraphics[width=0.3\textwidth,height=0.256\textwidth]{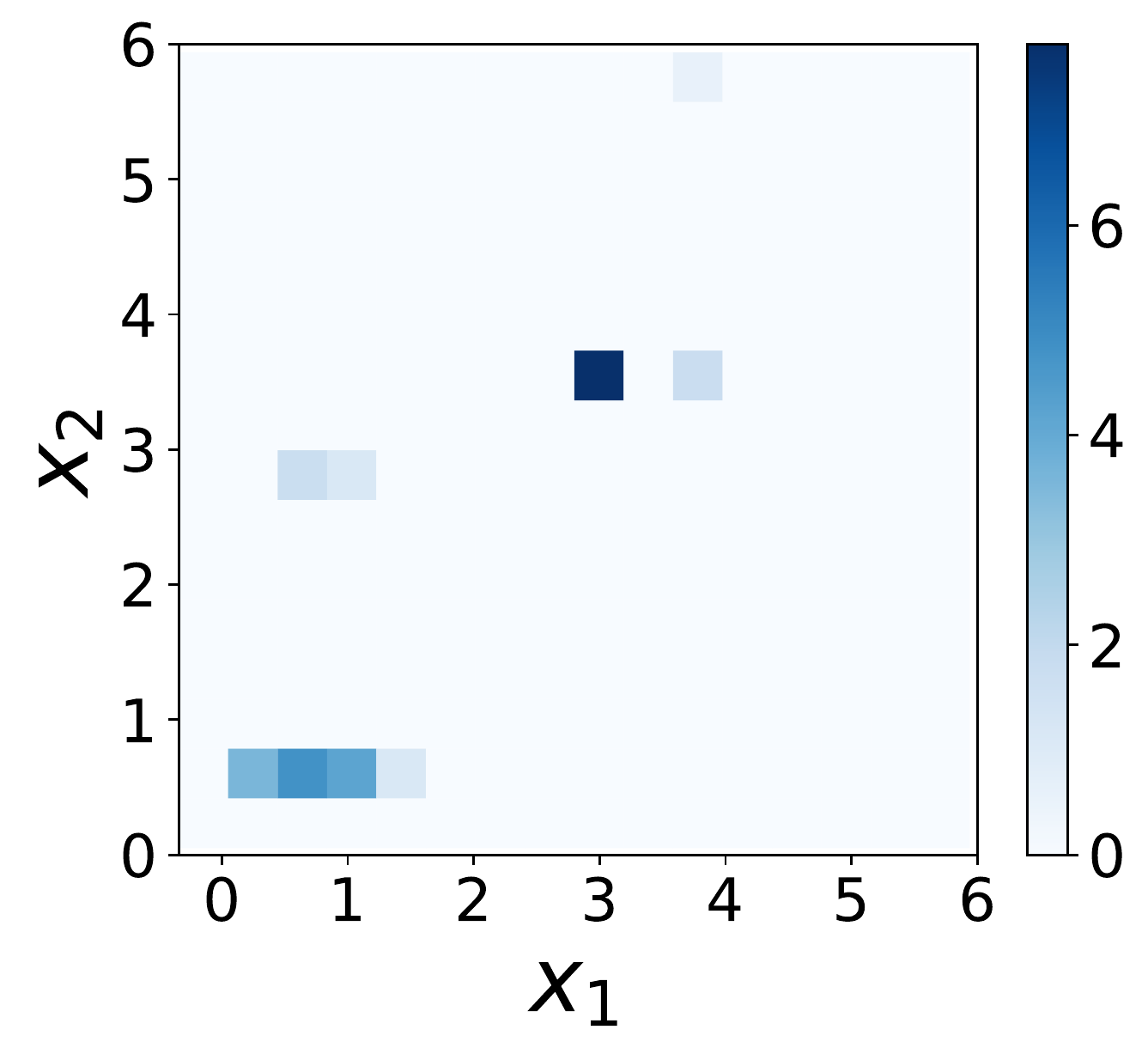}
    \caption{Wind direction data. Left: 2D histogram for wind directions; colorbar shows the counts of data points in each square. Mid: the 10 optimized locations (in red star), without repetition. Right: the objective value in Eq.\eqref{eq:mFFSD_obj}, $\frac{\mathrm{mFSSD^{2}}}{\tilde{\sigma}_{H_{1}}},$ on the specified data location of test (i.e. setting J=1); the higher the darker. 
    }
    \label{fig:wind}
\end{figure}

\subsection{Wind direction data}\label{sec:wind_direction}
As a real data on torus, we consider wind direction in Tokyo on 00:00 ($x_1$) and 12:00 ($x_2$) for each day in 2018\footnote{available on Japan Meteorological Agency website}.
Thus, the sample size is $n=365$.
The data were discretized into 16 directions, such as north-northeast.
Figure~\ref{fig:wind} presents a $16 \times 16$ histogram of raw data.


We consider gooodness-of-fit testing of the bivariate von Mises distribution in Eq.\eqref{bvM} via mKSD\footnote{We used the product kernel of the von Mises kernels $k((x_1,x_2),(y_1,y_2))=\exp (\eta_1 \cos(x_1-y_1) + \eta_2 \cos(x_2-y_2))$, where the parameters $\eta_1$ and $\eta_2$ were chosen by optimizing the approximate test power}.
By using noise contrastive estimation \cite{gutmann}, \cite{uehara2020imputation} fitted the bivariate von Mises distribution to the wind direction data and obtained the estimate
\[
\widehat{\xi}=(0.7170, 0.3954,1.1499, 1.1499,-1.1274).
\]
By setting this fitted model to the null model, the p-value by mKSD1 is 0.434, which indicates that the model fits data well.

In addition, we fitted a simpler model with no interactions between $x_1$ and $x_2$, i.e. $\lambda_{12}$ is set to zero in Eq.\eqref{bvM} so that the model reduces to the product of two von-Mises distribution on each direction.
The p-value by mKSD1 is 0.002, which is a strong evidence to reject the null model. 
In other words, there is a significant interaction between wind direction on 00:00 and 12:00.
We then carried out model criticism by mFFSD statistic in Eq.\eqref{eq:mFSSD} with optimized test location via maximizing approximate test power. 
Choosing the number of test locations $J=10$, we plot the optimized locations in Figure \ref{fig:wind}. 
It provides information about dependence between wind direction at midnight and noon.



\section*{Acknowledgement}
TM was supported by JSPS KAKENHI Grant Number 19K20220. WX was supported by Gatsby Charitable Foundation.

\bibliographystyle{apalike}
\bibliography{main}

\clearpage
\onecolumn
\appendix

\begin{center}
    \Large
        Supplementary Material for Interpretable Stein Goodness-of-fit Tests on Riemannian Manifolds \\
\end{center}
\section{Proofs and Derivations}

\subsection*{Stein's Identity}
Proof of Theorem \ref{thm:stein1_global_coord}
\begin{proof}
Let
$
\omega = \sum_{i=1}^d f^i \d \theta^{(-i)},
$
where 
$$\d \theta^{(-i)} =\d \theta^{i+1} \wedge \cdots \wedge \d \theta^{d} \wedge \d \theta^1 \cdots\ \wedge \d \theta^{i-1}$$ for $i=1,\dots,d$.
Then, 
\begin{align*}
\d(qJ\omega ) &= \sum_{i=1}^d \left( \frac{\partial f^i}{\partial {\theta}^i} + f^i \frac{\partial}{\partial {\theta}^i} \log (qJ) \right)\d \theta^{1} \wedge \cdots \wedge \d \theta^{d} =(qJ\mathcal{A}^{(1)}_{q} \mathbf f ) \d \theta^{1} \wedge \cdots \wedge \d \theta^{d}.
\end{align*}
Therefore, from Theorem~\ref{thm:stoke's} and Corollary~\ref{cor:stokes},
$$
\E_q[\mathcal{A}^{(1)}_{q} \mathbf f ] = \int_{\M}d(qJ\omega )=0.
$$
\end{proof}
\subsection*{Quadratic form of mKSD}
Proof of Theorem \ref{thm:quadratic_form}
\begin{proof} We show that, the mKSD admits the form of taking expectation over $p$ for bivariate functions $h^{(c)}_q$ which is independent of $p$. $h^{(c)}_q$ is also referred as the Stein kernel. The proof utilize the reproducing property of relevant RKHS and the fact that $\A_q^{(c)}$ is a linear functional of relevant test function $f$.

For $c=1$, the test function is a stack of $d$-dimensional RKHS functions $\mathbf{f}\in \H^{
d}$.
${\E}_p[\mathcal{A}^{(1)}_q \mathbf{f}]$ is a linear functional of $\mathbf{f} \in \mathcal{H}^{d}$.
Then, from the Riesz representation theorem, there uniquely exists $\mathbf{r}=(r_1,\dots,r_{d}) \in \mathcal{H}^{d}$ such that ${\E}_p[\mathcal{A}^{(1)}_q \mathbf{f}] = \langle \textbf{f}, \textbf{r}\rangle_{\mathcal{H}^{d}}$.
By using the reproducing property of $\mathcal{H}$ associate with kernel $k$, we obtain
\begin{align}
    r_i(x) = {\E}_{\tilde{x} \sim p} \left[ k(x,\tilde{x}) \frac{\partial}{\partial \tilde{\theta}^i} \log (qJ) + \frac{\partial}{\partial \tilde{\theta}^i} k(x,\tilde{x}) \right], \label{eq:def_g}
\end{align}
for $i=1,\dots,d$.
Thus, the maximization in $\operatorname{mKSD}^{(1)}(p\|q)$ is attained by $\mathbf{f}=\mathbf{r}/\|\mathbf{r}\|_{\mathcal{H}^{d}}$ and $\operatorname{mKSD}^{(1)}(p\|q)^2 =\| \mathbf{r} \|^2_{\mathcal{H}^{d}}$.
Therefore, the quadratic form is obtained after straightforward calculations: 
\begin{align*}
\operatorname{mKSD}^{(1)}(p\|q)^2 
&= \left\langle 
\E_{x\sim p}[\A_q^{(1)}k(x,\cdot)], \E_{\tilde{x}\sim p}[\A_q^{(1)}k(\tilde{x},\cdot)]\right\rangle_{\H^d} \\ 
& =\E_{x,\tilde{x}\sim p}
 \underset{h_q^{(1)}(x,\tilde{x})}{\left[\underbrace{\left\langle 
\A_q^{(1)}k(x,\cdot), \A_q^{(1)}k(\tilde{x},\cdot)\right\rangle_{\H^d}}\right]},
\end{align*}
and the assertion follows.

For $c=2$, similar argument applies where the test function is a scalar-valued RKHS $\tilde{f} \in \H$. Instead of Eq.\eqref{eq:def_g}, we have $\tilde{r} \in \H$, s.t. ${\E}_p[\mathcal{A}^{(2)}_q \mathbf{f}] = \langle \tilde{f}, \tilde{r}\rangle_{\mathcal{H}}$ and 
\begin{align}
    \tilde{r}(x) = {\E}_{\tilde{x} \sim p} \left[ \sum_{ij} g^{ij}\left( \frac{\partial}{\partial \tilde{\theta}^j}k(x,\tilde{x}) \frac{\partial}{\partial \tilde{\theta}^i} \log (qJ) + \frac{\partial^2}{\partial \tilde{\theta}^i\partial \tilde{\theta}^j} k(x,\tilde{x}) \right) \right], \label{eq:def_g2}
\end{align}
and the maximization in $\operatorname{mKSD}^{(2)}(p\|q)$ is attained by $\tilde{f}=\tilde{r}/\|\tilde{r}\|_{\mathcal{H}}$; thus $\operatorname{mKSD}^{(2)}(p\|q)^2 =\| \tilde{r} \|^2_{\mathcal{H}}$. The assertion then follows from the similar calculations as above.

For $c=0$, the quadratic form is readily obtained from derivation of maximum-mean-discrepancy (MMD) \cite{gretton2007kernel} form as shown in Theorem \ref{thm:mmd_equiv}. Alternatively, for scalar test function $h\in \H$, we can write,
$$
\operatorname{mKSD}^{(0)}(p\|q) =  \sup_{\|h \|_{\H} \leq 1} \mathbb{E}_{p}[\mathcal{A}^{(0)}_q h] =  \sup_{\|h \|_{\H} \leq 1} \left|\mathbb{E}_{p}[h] - \mathbb{E}_{q}[h]\right|,
$$
where taking the supreme we get,
\begin{align*}
\operatorname{mKSD}^{(0)}(p\|q)^2 &=  \Big\langle {{\mathbb{E}_{p}\big[k(x,\cdot) - \mathbb{E}_{q}[k(x,\cdot)]\big]}},\mathbb{E}_{p}\big[k(\tilde{x},\cdot) - \mathbb{E}_{q}[k(\tilde{x},\cdot)]\big]\Big\rangle_{\H}\\
&=\mathbb{E}_{x,\tilde{x}\sim p}\Big\langle \underset{\A_q^{(0)}k(x,\cdot)}{\underbrace{k(x,\cdot) - \mathbb{E}_{q}[k(x,\cdot)]}}, k(\tilde{x},\cdot) - \mathbb{E}_{q}[k(\tilde{x},\cdot)]\Big\rangle_{\H}.
\end{align*}
The assertion follows.
\end{proof}
The quadratic form is useful when computing the empirical estimate for the expectation where only samples from unknown distribution $p$ is observed. We also note that $ \mathbb{E}_{q}[k(\tilde{x},\cdot)]$, in general, is not possible to obtain in analytical form, especially when the density $q$ is only given up to normalization. Samples from $q$, if possible to obtain from unnormalized density, can be useful to estimate $\A_q^{(0)}k(x,\cdot)$, where we denote as $\widehat{\A_q^{(0)}}k(x,\cdot)$.

\subsection*{Characterisation of mKSD}
Proof of Theorem \ref{thm:characteristic}
\begin{proof}
Denote $\textbf{s}^{(c)}_p(\cdot) = \E_{\tilde{x}\sim p}[\A_q^{(c)} k(\tilde{x},\cdot)]\in \mathcal{F}$ and we can write 
$$\operatorname{mKSD}^{(c)}(p\|q)^2=\|\textbf{s}_p(\cdot)\|^2_{\mathcal F} \geq 0,$$ 
where $\mathcal{F}$ can be $\H$ for $c=0,2$ or $\H^d$ for $c=1$. 
If $p=q$, then $\operatorname{mKSD}^{(c)}(p\|q)^2=0$ from the Stein identity. 

Conversely, if $\operatorname{mKSD}^{(c)}(p\|q)^2=0$, then $\textbf{s}^{(c)}_p(x)=\textbf{0}$, a zero vector in $\R^d$ for $c=1$ and a scalar zero in $\R$ for $c=0,2$, $\forall x$, s.t. $p(x)>0$.
Then, from $\log (q/p) = \log (qJ) - \log (pJ)$, we obtain,
\begin{align*}
    &{\E}_{\tilde{x} \sim p} \left[ L_i(\tilde{x}) k(\tilde{x},x) \right] = (\textbf{s}^{(1)}_p)_i(x) - {\E}_{\tilde{x} \sim p} \left[ \mathcal{A}^{(1)}_p k(\tilde{x},x)  \right] = 0,
\end{align*}
and 
\begin{align*}
    {\E}_{\tilde{x} \sim p} \left[ L(\tilde{x}) k(\tilde{x},x) \right] = (\textbf{s}^{(c)}_p)(x) - {\E}_{\tilde{x} \sim p} \left[ \mathcal{A}^{(c)}_p k(\tilde{x},x)  \right] = 0,
\end{align*}
for $c=0,2$,
for every $x$ with positive densities.
Since $k$ is compact-universal, vanishes at $\partial \M$ and $\M$ is smooth and compact, the injectivity result in \cite[Theorem 4(b)]{carmeli2010vector} implies that $L^{(1)}_i=0, \forall i$  (for $c=1$, $i \in \{1,\dots,d\}$; for $c=0,2$, $i=1$).
Therefore, $\log (q/p)$ is constant on $\M$.
Since both $p$ and $q$ are both densities on $\M$ that integrate to one, we conclude $p=q$.
\end{proof}

\subsection*{Asymptotics of mKSD}
Proof of Theorem \ref{thm:u-stat-asymptotic}
\begin{proof}
To show part 1, it is enough to check the mKSD statistics is degenerate U-statistics under $H_0: p=q$. By considering test function $f = k(x,\cdot)$ (or its relevant vector-valued form for $c=1$), Stein identity shows that, 
\begin{align*}\label{eq:mksd_one_x}
    \mathbb{E}_{\tilde{x} \sim p} [\A^{(c)}_q k(x,\tilde{x})] = 0, \forall x \in \M,
\end{align*}
so that the variance $\sigma^2_c = 0$ for $c=0,1,2$. Then the standard results for degenerate U-statistics in \cite[Section 5.5.2]{serfling2009approximation} apply and the assertions follow. 

In addition, it is interesting to note link the result for $c=0$ with the asymptotic result in  as 
$$h^{(0)}_q(x, \tilde{x}) = k(x,\tilde{x}) -  \xi(x) - \xi(\tilde{x}) + C,$$ 
where $C = \E_{x,\tilde{x}\sim q}k(x,\tilde{x})$ is a constant, $\xi(x) = \E_{\tilde{x}\sim q}k(x,\tilde{x})$ is only a function of $x$ and $\xi(\tilde{x}) = \E_{{x}\sim q}k(x,\tilde{x})$ is only a function of $\tilde{x}$. The formulation is analogous to the asymptotic results for MMD, as shown in \cite[Theorem 8]{gretton2007kernel}: $h_q^{(0)}(x,\tilde{x})$ is equivalent to the notion of $\tilde{k}(x,\tilde{x})$ in \cite{gretton2007kernel}.

Part 2 follows as $\sigma_c^2 >0$ under $H_1: p\neq q$ by Theorem \ref{thm:characteristic}. Apply asymptotic distribution of non-degenerate U-statistics \cite[Section 5.5.1]{serfling2009approximation} and the assertions follow.
\end{proof}

\paragraph{Asymptotics for mFSSD} To compute the empirical version of mFSSD, we consider the empirical version $\mathbf{s}_p(\cdot)$ in Eq.\eqref{eq:mFSSD} from samples $x_1, \dots, x_n \sim p$: 
$$\widehat{\mathbf{s}}_p(\cdot) = \frac{1}{n}\sum_i [\A^{(1)}_q k(x_i,\cdot)].$$ 
Then the empirical mFSSD has the form
\begin{equation}\label{eq:mfssd_hat}
    \widehat{\operatorname{mFSSD}^2} = \frac{1}{dJ}\sum_{i=1}^d\sum_{j=1}^J (\widehat{\textbf{s}}_p(\mathbf{v}_j))_i^2,
\end{equation}
for any set of test locations $\{v_j\}_{j=1}^J$.
\begin{Proposition}\label{prop:fssd_asym}
Assume the conditions in Theorem \ref{thm:characteristic} hold, and $\E_{x\sim p}[\|\mathbf{s}_p(x)\|^2]<\infty$. Under $H_1: p\neq q$, $$\sqrt{n}\cdot\left(\widehat{\operatorname{mFSSD}^2}-\operatorname{mFSSD}^2\right) \overset{d}{\to} \mathcal{N}(0,\tilde{\sigma}^2_{H_1}),$$
where $\tilde{\sigma}^2_{H_1}$ denotes the variance for $\widehat{\operatorname{ mFSSD}^2}$. 
\end{Proposition}
\begin{proof}
With the assumed regularity conditions, Eq.\eqref{eq:mfssd_hat} is in the form of the non-degenerate U-statistics with $\tilde{\sigma}^2_{H_1} >0$. The asymptotic normality follows from \cite[Section 5.5.1]{serfling2009approximation}, similarly described in \cite[Proposition 2]{jitkrittum2017linear}.
\end{proof}

The asymptotic normality for $\widehat{\operatorname{mFSSD}^2}$ in Proposition \ref{prop:fssd_asym} enables derivation of the approximate test power, similarly as described in Section \ref{sec:gof_procedure} for kernel choice. 
\begin{Proposition}\label{prop:mfssd-power}
\textnormal{[Approximate test power of $n\cdot\operatorname{\widehat{{mFSSD^{2}}}}$]} Under $H_{1}$, for large $n$ and fixed $r$, the test power is
$$\mathbb{P}_{H_{1}}(n\cdot\widehat{\mathrm{mFSSD^{2}}}>r)\approx1-\Phi\left(\frac{r}{\sqrt{n}\tilde{\sigma}^2_{H_1}}-\sqrt{n}\frac{\mathrm{mFSSD^{2}}}{\tilde{\sigma}^2_{H_1}}\right),$$ 
where $\Phi$ denotes the cumulative distribution function of the standard normal distribution, and $\tilde{\sigma}^2_{H_1}$ is defined in Proposition \ref{prop:fssd_asym}.
\end{Proposition}

Due to $\sqrt{n}$ scaling in Proposition \ref{prop:fssd_asym}, maximizing the approximate test power for $n\cdot\operatorname{\widehat{{mFSSD^{2}}}}$ can be approximated by maximizing $\frac{\operatorname{mFSSD}^2}{\tilde{\sigma}^2_{H_1}}$ to obtain optimal test locations under the alternative $H_1: p \neq q$, which  
is described in Section \ref{sec:model_criticism}.
$$
V = \argmax_{\mathbf{v}} \frac{\mathrm{mFSSD^{2}}}{\tilde{\sigma}_{H_{1}}},
$$
for $V=\left\{ \mathbf{v}_{j}\right\}_{j=1}^{J}$ as the set of test locations to be optimised.

\section{More on Bahadur Efficiency}
In this section, we introduce the relevant concepts to study Approximate Relative Efficiency (ARE) between two tests, characterised by \emph{Bahadur slope} \cite{bahadur1960stochastic} and corresponding \emph{Bahadur efficiency}. 


\subsection{Approximate Bahadur Slope}
We first define Bahadur slope for general tests \cite{gleser1966comparison} and its applications in kernel-based tests \cite{jitkrittum2017linear,garreau2017large}.
Consider the test procedure with null hypothesis $H_0:\omega \in \Omega_0$ and the alternative $H_1:\omega \in \Omega\backslash\Omega_0$, where $\Omega$ and $\Omega_0$ are arbitrary sets. Denote $T_n$ as the test statistic computed from a sample of size $n$.
\begin{Definition}
For $\omega_0 \in \Omega_0$, let F be the asymptotic null distribution
$$
F(t) = \lim_{n\to \infty} \P_{\omega_0}(T_n < t)
$$
which is assumed to be continuous and common $\forall \omega_0 \in \Omega_0$.
Assume that there exists a continuous strictly increasing function $\rho: (0,\infty) \to (0,\infty)$ s.t $\lim_{n\to \infty} \rho(n) = \infty$. Denote 
\begin{align}\label{eq:bahadur_slope}
    c(\omega) = -2\operatorname{plim}_{n\to \infty}\frac{\log (1-F(T_n))}{\rho(n)},
\end{align}
for some bounded non-negative function $c$ such that $c(\omega_0)=0$ when $\omega_0\in \Omega_0$. The function $c(\omega)$ is known as \emph{approximate Bahadur slope}.
\end{Definition}

\begin{Definition}
Let $\mathcal{D}(a,t)$ be a class of all continuous cumulative distribution functions (CDF) F such that $-2\log(1-F(x))=ax^{t}(1+o(1))$, as $x\to\infty$ for $a>0$ and $t>0$.
\end{Definition}
\begin{Proposition}\label{thm:abs}
The approximate Bahadur slope (ABS) for the tests with $\operatorname{mKSD^{(c)}}$, $c=0,1,2$ is
$$c^{({\operatorname{mKSD^{(c)}}})}:=\frac{\E_p[h^{(c)}_q(x,\tilde{x})]}{\E_q[h^{(c)}_q(x,\tilde{x})^2]^{\frac{1}{2}}},$$ 
where $h^{(c)}_q(x,\tilde{x})$ is the Stein kernel for $\operatorname{mKSD^{(c)}}$, and  $\rho(n)=n$.
\end{Proposition}
\begin{proof}
Using Theorem 9 and Theorem 11 in \cite{jitkrittum2017linear}, we know that $n\cdot \operatorname{mKSD}_u^{(c)}(p\|q)^2$ in Eq.\eqref{eq:u-stats} is in the class of $\mathcal{D}(a=1/\omega_c,t=1)$ for $\omega^2_c$ is the variance of the statistic. By Stein identity, $\E_{x\sim q} \E_{\tilde{x}\sim q}\left[h^{(c)}_q(x,\tilde{x})\right]^2 = 0$. Hence, using second point in Theorem 9 \cite{jitkrittum2017linear} and choosing $\rho = n$, we know that $n\cdot \operatorname{mKSD}_u^{(c)}(p\|q)^2\backslash \rho(n) \to \operatorname{mKSD}^{(c)}(p\|q)^2$ by weak law of large numbers. 
\end{proof}
 
\subsection{Asymptotic Relative Efficiencies Between mKSD Tests with Different $\A_q$s}

Asymptotic Relative Efficiency (ARE) between two statistical testing procedures measures how fast the p-values of one test shrinks to $0$, relatively to the other's. If it is faster, for given problem under the alternative, it is more sensitive to pick up the alternative, where we call the test more "statistically efficient". With 
ABS, we are ready to define approximate Bahadur efficiency.

\begin{Definition}
Given two sequences of test statistics, $T_{n}^{(1)}$ and $T_{n}^{(2)}$ 
and their ABS $c^{(1)}$ and $c^{(2)}$,  the  
{approximate Bahadur efficiency} of $T_{n}^{(1)}$ relative to $T_{n}^{(2)}$ is 
\begin{align}\label{eq:bahadur_efficiency}
{\rm E}(\omega_{A}):=\frac{c^{(1)}(\omega_{A})}{c^{(2)}(\omega_{A})}
\end{align}
for $\omega_{A}\in \Omega \backslash \Omega_{0}$, in the space of alternative models. 
\end{Definition}
If $\rm E(\omega_{A})>1$, then $T_{n}^{(1)}$ is asymptotically more efficient than $T_{n}^{(2)}$ in the sense of Bahadur, for the particular problem specified by $\omega_{A}\in\Omega\backslash\Omega_{0}$.

\subsection{The Case Study on Circular distribution $\Sp^1$}
{Proof of Theorem \ref{thm:efficienty_12}}
\begin{proof}
To compute $\rm{E}_{1,2}(\kappa)$, we can rewrite the following:
$$
{\rm{E}}_{1,2}(\kappa) = \frac{\E_p[h^{(1)}_q(x,\tilde{x})]}{\E_p[h^{(2)}_q(x,\tilde{x})]}\cdot \frac{\E_q[h^{(2)}_q(x,\tilde{x})^2]^{\frac{1}{2}}}{\E_q[h^{(1)}_q(x,\tilde{x})^2]^{\frac{1}{2}}}
$$
The second term only involves integrals over $q(x) \propto 1$, which is independent of $\kappa$ and we can solve it as $\frac{\E_q[h^{(2)}_q(x,\tilde{x})^2]^{\frac{1}{2}}}{\E_q[h^{(1)}_q(x,\tilde{x})^2]^{\frac{1}{2}}} = 1.692>1$. 
For the first term, the ratio is monotonic decreasing w.r.t. $\kappa>0$ and 
$\frac{\E_p[h^{(1)}_q(x,\tilde{x})]}{\E_p[h^{(2)}_q(x,\tilde{x})]} $ is lower bounded by $2$ due to exponential-trace kernel and $\Sp^1$ embedded in $\R^2$. 
Hence, for $\kappa > 0$, ${\rm E}_{1,2} > 1$.
\end{proof}
\begin{wrapfigure}{r}{0pt}
{\includegraphics[width=0.4\textwidth,height=0.26\textwidth]{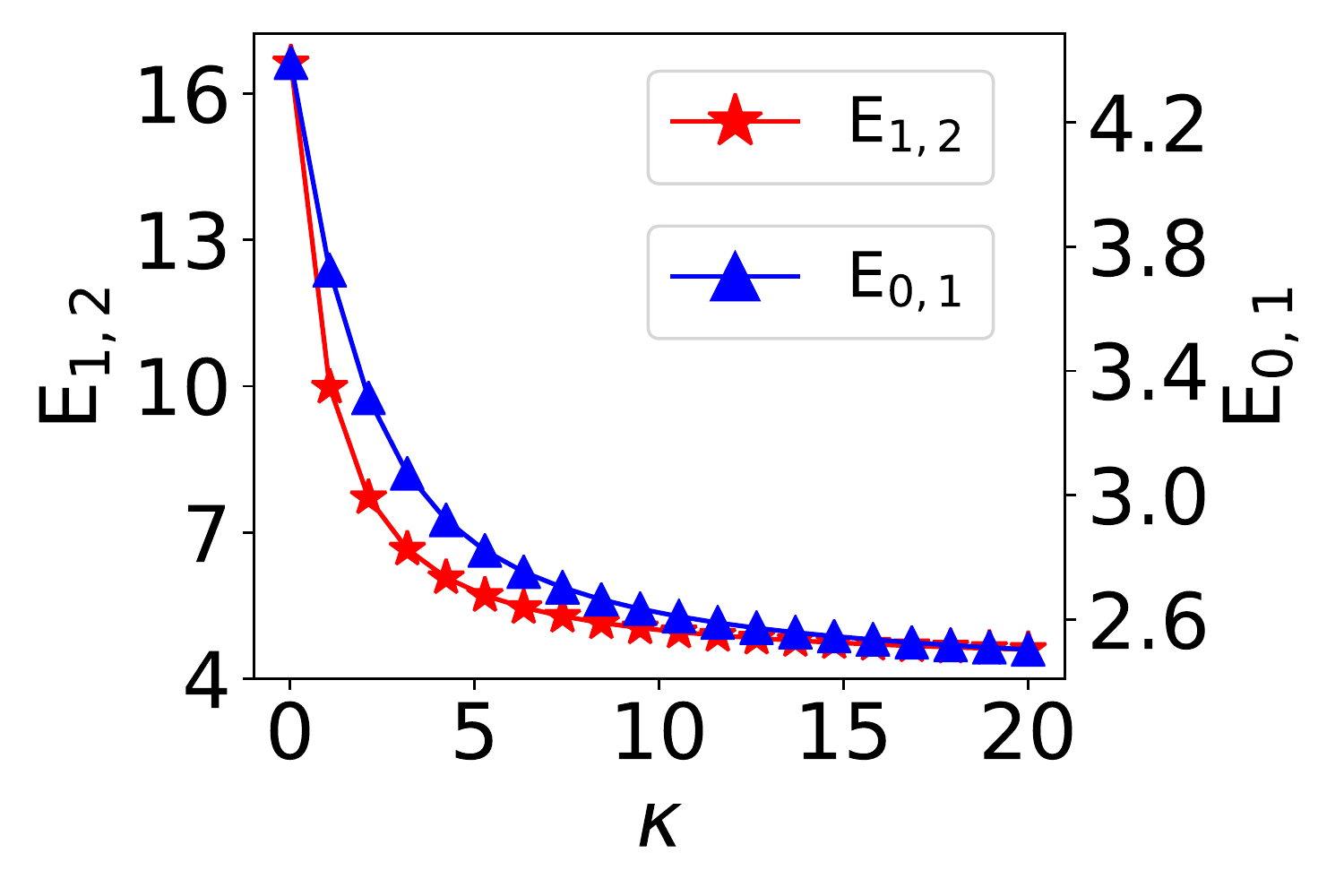}}
\vspace{-2cm}
\caption{Relative Test Efficiency}\label{fig:efficiency_illustration}
\vspace{-1.5cm}
\end{wrapfigure}
We can apply similar approach to compare the relative test efficiency ${\rm E}_{0,1}(\kappa)$ between 
$\operatorname{mKSD}^{(0)}$ and $\operatorname{mKSD}^{(1)}$. We plot numerical solutions in Figure \ref{fig:efficiency_illustration}. 
From Figure \ref{fig:efficiency_illustration}, we see that $\E_{1,2}$ and $\E_{0,1}$ both greater than $1$ for $\kappa \in (0,20)$. 
For further increase of $\kappa$, there is a trend for both relative efficiencies stabilizing at some value greater than $1$. Theoretical analysis for such limiting behaviour is of an interesting future topic.  
Although Figure \ref{fig:efficiency_illustration} shows that $\E_{0,1}(\kappa)>1$ for small perturbation from the null, i.e. $\kappa\in(0,20)$ which suggest the relative efficiency of $\operatorname{mKSD}^{(0)}$ is higher than the first order test $\operatorname{mKSD}^{(1)}$, it is usually not possible to compute MMD analytically and the normalized density is required. 

Intuitively, with sampling error of order $\sqrt{n}$ and $\rho(n)=n$ is chosen to compute Bahadur slope, the MMD computed from samples are less efficient to perform goodness-of-fit test compared to mKSD tests that directly access the unnormalized density, as shown in Figure \ref{fig:synthetic-problems}. Similar findings are also observed in other settings where MMD is considered to perform goodness-of-fit tests
\cite{liu2016kernelized,jitkrittum2017linear,yang2018goodness,yang2019stein, xu2020stein}.  
In addition, correctly sampling from Riemannian manifold is non-trivial and can be time-consuming for sample-based tests.

\section{More on Model Criticism}
In this section, we provide additional details on model criticism for wind data present in Section \ref{sec:wind_direction}. 
We fitted the model in Eq.\eqref{bvM} by using noise contrastive estimation \cite{uehara2020imputation} and our test does not find evidence to reject the fitted model, suggesting a good fit for the wind direction data. In addition, we consider the model without interaction term between two direcitons:
\begin{align}
\widetilde{q}(x_1,x_2 \mid \widetilde{\xi}) \propto & \exp \{\kappa_1 \cos(x_1-\mu_1) + \kappa_2 \cos(x_2-\mu_2)\}, \label{eq:product_vmf}
\end{align}
which is equivalent to model in Eq.\eqref{bvM} by imposing $\lambda_{12}=0$. This model can be viewed as product of marginal distributions of $x_1$ and $x_2$ and we refer as factorized model. Our test reject the null at test level $\alpha=0.05$ suggesting a poor fit of the factorized model. 

To further visualize the difference between models in Eq.\eqref{bvM} and Eq.\eqref{eq:product_vmf}, we plot histogram of each wind direction in Figure \ref{fig:wind_1d_marginal} and samples from the factorized model $\widetilde{q}$ in Figure \ref{fig:wind_2d_marginal} where 
no interactions are present between $x_1$ and $x_2$. Compare with the wind direction data, shown again in Figure \ref{fig:wind_joint}, we can see that Figure \ref{fig:wind_2d_marginal} differs the most at the regions of  $\tilde{x} = (x_1, x_2)=(2.8, \pi)$ (data model denser) and $\tilde{x}'=(x_1, x_2)=(1,1)$ ($\tilde{q}$ model denser). Such difference is captured by our optimized test locations from mFSSD in Figure \ref{fig:opt_J10}, where $\tilde{x}$ is at the region with 3 stars in a row and $\tilde{x}'$ is around the region with 4-stars in a row. It shows the effectiveness of mFSSD in distinguishing the differences between distributions. 
As $\tilde{q}$ is referred as imposing data model in Eq.\eqref{bvM} to be $0$,
a negative  $\lambda_{12} = -1.1274<0$ in the data model implies that positive $\sin(x_1 - \mu_1)\sin(x_2 - \mu_2)$ is less dense. With $\mu_1=1.1499=\mu_2$, $\sin(x_1 - \mu_1)\sin(x_2 - \mu2)$ is positive around the region the $\tilde{x}'$ making the data model less dense, as shown in Figure \ref{fig:wind_joint} and \ref{fig:wind_2d_marginal}.

\begin{figure*}[t!]
    \centering
    \subfigure[Wind direction model\label{fig:wind_joint}]{\includegraphics[width=0.312\textwidth,height=0.26\textwidth]{wind_direction_hist.pdf}}
    \subfigure[Marginals for each direction\label{fig:wind_1d_marginal}]{\includegraphics[width=0.32\textwidth,height=0.256\textwidth]{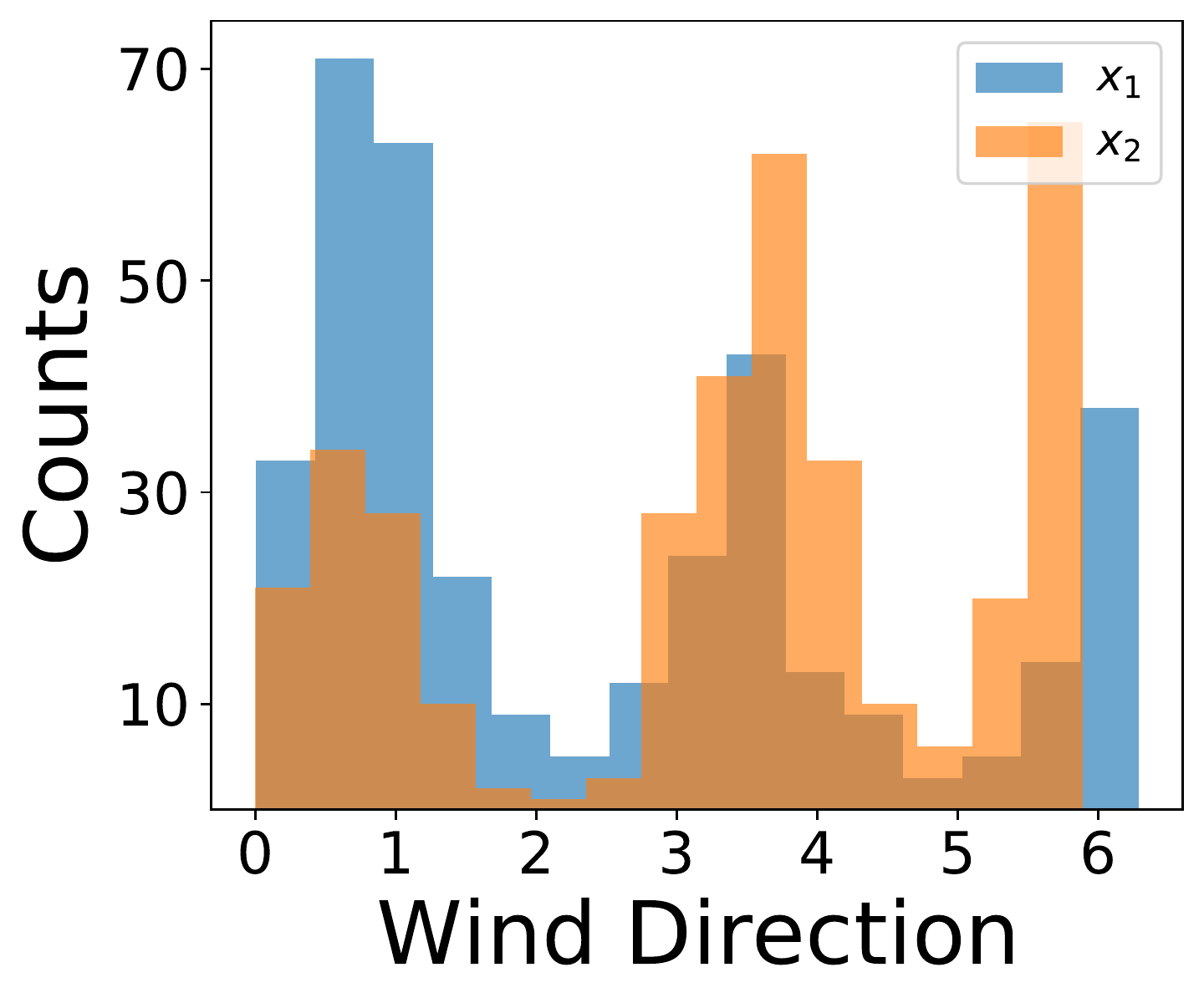}}
        \subfigure[Samples from the factorized model\label{fig:wind_2d_marginal}]{\includegraphics[width=0.312\textwidth,height=0.26\textwidth]{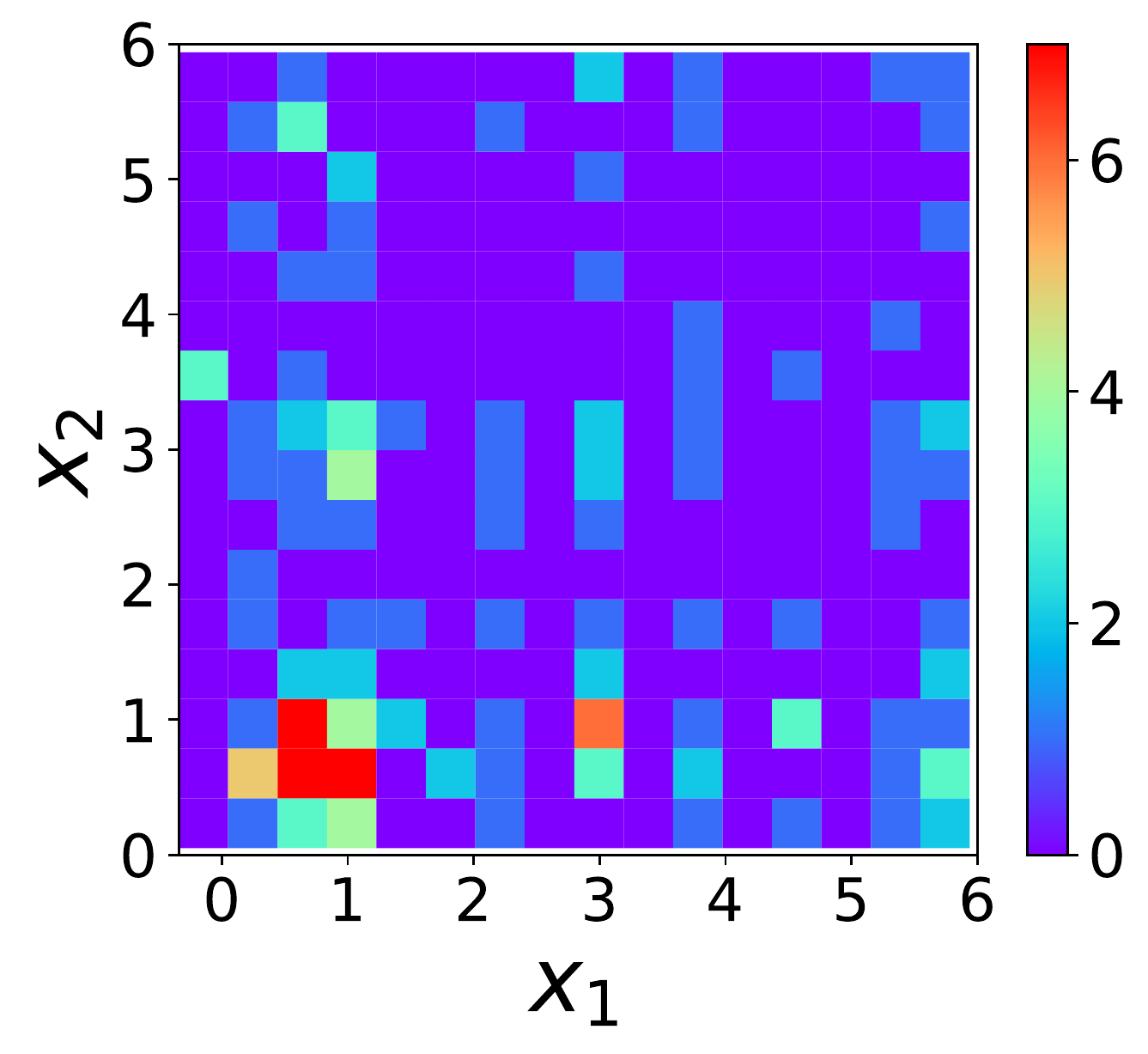}}

        \subfigure[Objective values for $J=10$ \label{fig:opt_loc_hist}]{\includegraphics[width=0.312\textwidth,height=0.26\textwidth]{wind_opt_loc_hist.pdf}}\hspace{0.03\textwidth}
        \subfigure[Optimized locations, for $J=10$\label{fig:opt_J10}]{\includegraphics[width=0.28\textwidth,height=0.256\textwidth]{wind_direction_J10.pdf}}\hspace{0.03\textwidth}
        \subfigure[Objective values for a single test location \label{fig:opt_loc_value}]{\includegraphics[width=0.28\textwidth,height=0.256\textwidth]{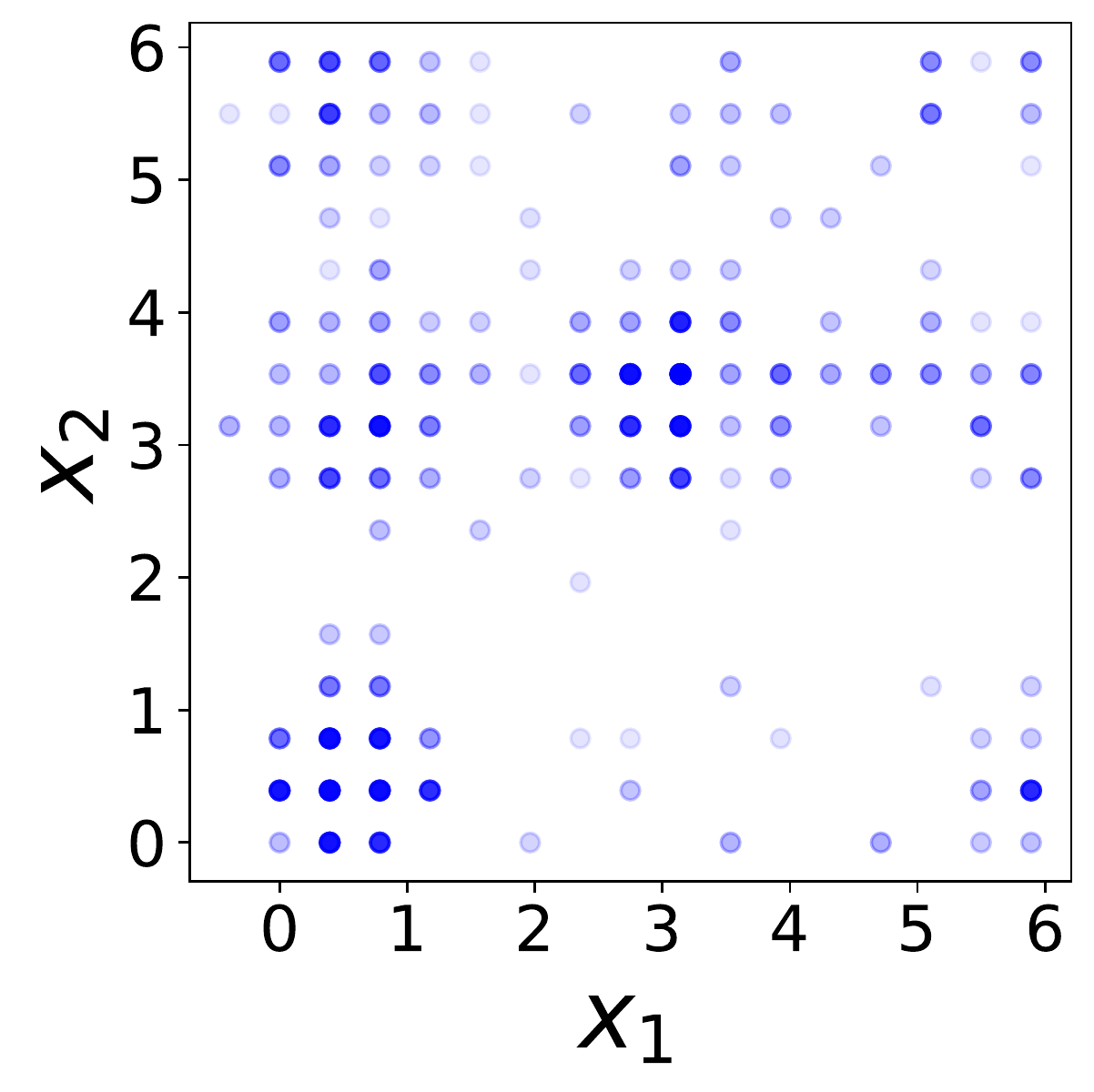}}
    \caption{Visualizing the fitted model and rejected model for wind direction data.
    }
    \label{fig:wind_interp}
\end{figure*}

\end{document}